\documentclass{article}
\pdfoutput=1
\usepackage[utf8]{inputenc}
\usepackage{todonotes}
\usepackage{hyperref}
\usepackage{bbm}
\usepackage{amsfonts}
\usepackage{amsmath}
\usepackage{amssymb}
\usepackage{mathtools}
\usepackage{amsthm}
\usepackage{cleveref}
\usepackage{xfrac}
\usepackage{float}
\usepackage{soul,color}
\usepackage{algorithm}
\usepackage{algpseudocode}
\usepackage[title]{appendix}
\usepackage[letterpaper]{geometry}
\usepackage{xspace}
\usepackage{multirow}
\usepackage{comment}
\usepackage{stackengine}
\usepackage{xcolor}
\usepackage{changepage}
\usepackage{diagbox}
\usepackage{authblk}
\usepackage{titling}

\usetikzlibrary {patterns,patterns.meta} 
\usetikzlibrary{fit}
\usetikzlibrary{arrows, arrows.meta}
\usetikzlibrary{shapes}
\usetikzlibrary{decorations.pathreplacing}
\usetikzlibrary{calc}

\title{Counting Permutation Patterns with Multidimensional Trees}

\newif\ifblind

\ifblind

\author{
\vphantom{A}\vspace{0.1cm}
\\\small \vphantom{B}
}

\else

\author{
Gal Beniamini \thanks{\href{mailto:gal.beniamini@mail.huji.ac.il}{gal.beniamini@mail.huji.ac.il}} \quad Nir Lavee \thanks{\href{mailto:nir.lavee@mail.huji.ac.il}{nir.lavee@mail.huji.ac.il}}
\\ \small The Hebrew University of Jerusalem
}

\fi

\date{}

\hypersetup{colorlinks=true}	

\newtheorem{theorem}{Theorem}[section]
\newtheorem*{theorem*}{Theorem}

\newtheorem{lemma}[theorem]{Lemma}
\newtheorem{proposition}[theorem]{Proposition}

\newtheorem{definition}[theorem]{Definition}

\newtheorem*{openq*}{Question}
\newtheorem{exmp}{Example}[section]

\newcommand{\Sn}{\mathbb{S}_n}

\newcommand{\QQ}{\mathbb{Q}}
\newcommand{\ZZ}{\mathbb{Z}}

\newcommand{\pc}[2]{{\# \mathtt{ #1 } \left( #2 \right)}}
\newcommand{\pcn}[2]{{\# #1 } \left( #2 \right)}
\newcommand{\pce}[1]{{\# \mathtt{#1} }}

\newcommand{\pcw}[2]{{\#_w \mathtt{ #1 } \left( #2 \right)}}
\newcommand{\pcwe}[1]{{\#_w \mathtt{ #1 } }}

\newcommand{\pcnw}[2]{{\#_w { #1 } \left( #2 \right)}}

\newcommand{\eqdef}{\vcentcolon=}

\DeclareMathOperator{\polylog}{polylog}

\DeclareMathOperator{\poly}{poly}

\theoremstyle{remark}
\newtheorem{rem}[theorem]{\protect\remarkname}
\providecommand{\remarkname}{Remark}

\newenvironment{proofsketch}{%
  \proof}{\endproof}

\newcommand*\circled[1]{\tikz[baseline=(char.base)]{
            \node[shape=circle,draw,inner sep=0.7pt] (char) {#1};}}

\def\acts{\curvearrowright}

\pgfdeclarepattern{
  name=hatch,
  parameters={\hatchsize,\hatchangle,\hatchlinewidth},
  bottom left={\pgfpoint{-.1pt}{-.1pt}},
  top right={\pgfpoint{\hatchsize+.1pt}{\hatchsize+.1pt}},
  tile size={\pgfpoint{\hatchsize}{\hatchsize}},
  tile transformation={\pgftransformrotate{\hatchangle}},
  code={
    \pgfsetlinewidth{\hatchlinewidth}
    \pgfpathmoveto{\pgfpoint{-.1pt}{-.1pt}}
    \pgfpathlineto{\pgfpoint{\hatchsize+.1pt}{\hatchsize+.1pt}}
    \pgfpathmoveto{\pgfpoint{-.1pt}{\hatchsize+.1pt}}
    \pgfpathlineto{\pgfpoint{\hatchsize+.1pt}{-.1pt}}
    \pgfusepath{stroke}
  }
}
\tikzset{
  hatch size/.store in=\hatchsize,
  hatch angle/.store in=\hatchangle,
  hatch line width/.store in=\hatchlinewidth,
  hatch size=5pt,
  hatch angle=0pt,
  hatch line width=.5pt,
}

\pgfdeclarepatternformonly{south west lines}{\pgfqpoint{-0pt}{-0pt}}{\pgfqpoint{6pt}{6pt}}{\pgfqpoint{6pt}{6pt}}{
    \pgfsetlinewidth{0.4pt}
    \pgfpathmoveto{\pgfqpoint{0pt}{0pt}}
    \pgfpathlineto{\pgfqpoint{6pt}{6pt}}
    \pgfpathmoveto{\pgfqpoint{5.6pt}{-.4pt}}
    \pgfpathlineto{\pgfqpoint{6.4pt}{.4pt}}
    \pgfpathmoveto{\pgfqpoint{-.4pt}{5.6pt}}
    \pgfpathlineto{\pgfqpoint{.4pt}{6.4pt}}
\pgfusepath{stroke}}

\pgfdeclarepatternformonly{south east lines}{\pgfqpoint{-0pt}{-0pt}}{\pgfqpoint{6pt}{6pt}}{\pgfqpoint{6pt}{6pt}}{
    \pgfsetlinewidth{0.4pt}
    \pgfpathmoveto{\pgfqpoint{0pt}{6pt}}
    \pgfpathlineto{\pgfqpoint{6pt}{0pt}}
    \pgfpathmoveto{\pgfqpoint{.4pt}{-.4pt}}
    \pgfpathlineto{\pgfqpoint{-.4pt}{.4pt}}
    \pgfpathmoveto{\pgfqpoint{6.4pt}{5.6pt}}
    \pgfpathlineto{\pgfqpoint{5.6pt}{6.4pt}}
\pgfusepath{stroke}}

\newcommand{\Otilde}[1]{\widetilde{\mathcal{O}}\left( #1 \right)}
\newcommand{\Oh}[1]{\mathcal{O}\left( #1 \right)}

\theoremstyle{plain}
\newtheorem{thm}{Theorem}
\newtheorem{introthm}{Theorem}

\begin{document}

\maketitle

\vspace{-1cm}

\begin{abstract}

We consider the well-studied \textit{pattern counting problem}: given a permutation $\pi \in \Sn$
and an integer $k > 1$, count the number of order-isomorphic occurrences
of every pattern $\tau \in \mathbb{S}_k$ in~$\pi$.

Our first result is an $\Otilde{n^2}$-time algorithm for $k=6$ and $k=7$.
The proof relies heavily on a new family of graphs that we introduce, called \textit{pattern-trees}.
Every such tree corresponds to an integer linear combination
of permutations in $\mathbb{S}_k$,
and is associated with linear extensions of partially ordered sets.
We design an evaluation algorithm for these combinations, and apply
it to a family of linearly-independent trees.
For $k=8$, we show a barrier:
the subspace spanned by trees in the previous family
has dimension exactly $|\mathbb{S}_8| - 1$, one less than required.

Our second result is an $\widetilde{\mathcal{O}}(n^{7/4})$-time algorithm for $k=5$. 
This algorithm extends the framework of pattern-trees by speeding-up their evaluation
in certain cases. A key component of the proof is the introduction 
of pair-rectangle-trees, a data structure for dominance counting.
\end{abstract}

\section{Introduction}
\label{sect:introduction}

A permutation $\tau \in \mathbb{S}_k$ occurs in a permutation $\pi \in \mathbb{S}_n$
if there exist $k$ points in $\pi$ that are order-isomorphic to $\tau$.
By way of example, in $\mathtt{\overline{1}3\overline{42}} \in \mathbb{S}_4$,\footnote{
Throughout this paper, permutations are written in one-line notation. If they are short, we omit the parenthesis.
}
the overlined points form an occurrence of $\mathtt{132} \in \mathbb{S}_3$.
The number of occurrences $\pc{\tau}{\pi}$ of a permutation 
$\tau \in \mathbb{S}_k$ (a \textit{pattern}) within a larger permutation $\pi \in \Sn$
has been the basis of many interesting questions, both combinatorial and algorithmic.

In a classical result, MacMahon \cite{macmahon1915combinatory} proved that the number of permutations $\pi \in \Sn$
that \textit{avoid} the pattern $\mathtt{123}$ (i.e., $\pc{\mathtt{123}}{\pi} = 0$) is counted by the
Catalan numbers. Another classical result is the well-known Erd\H{o}s-Szekeres theorem \cite{erdos1935combinatorial},
which states that any permutation of size $(s-1)(l-1) + 1$
cannot simultaneously avoid both $(\mathtt{1, \dots , s})$ \textit{and} $(\mathtt{l, \dots , 1})$.
These early results gave rise to an entire field of study regarding pattern avoidance,
c.f. \cite{pratt1973computing, knuth1997art, simion1985restricted}.
One particularly noteworthy result is Marcus and Tardos'
resolution of the Stanley-Wilf conjecture \cite{marcus2004excluded}:
for any fixed pattern $\tau \in \mathbb{S}_k$, the growth rate of the number of permutations
$\pi \in \mathbb{S}_n$ avoiding $\tau$ is $c(\tau)^n$, where $c(\tau)$ is a constant depending only on $\tau$.

Pattern avoidance can also be cast as an algorithmic problem.
The \textit{permutation pattern matching} problem is the task
of determining, given a pattern $\tau \in \mathbb{S}_k$ and a permutation $\pi \in \Sn$,
whether $\pi$ avoids $\tau$.
What is the computational complexity of this task? 
Trivial enumeration over all $k$-tuples of points yields an $\mathcal{O}(k \cdot n^k)$-time algorithm.
This bound has been improved upon by a long line of works:
Albert et al. \cite{albert2001algorithms} lowered the bound to $\mathcal{O}(n^{2k/3 + 1})$,
Ahal and Rabinovich \cite{ahal2008complexity} to $\mathcal{O}(n^{(0.47 + o(1))k})$,
and finally Guillemot and Marx \cite{guillemot2014finding} established the fixed-parameter tractability
of the problem, i.e., whenever $k$ is \textit{fixed},
the problem can be solved in time linear in $n$
(see also \cite{fox2013stanley} for an improvement on this result).
In stark contrast, when the pattern $\tau$ is \textit{not fixed} (i.e., when $k=k(n)\to \infty$),
permutation pattern matching is known to be NP-complete, as shown
by Bose, Buss and Lubiw \cite{bose1998pattern}.

A closely related algorithmic question is the \textit{counting version} of permutation pattern matching.
The \textit{permutation pattern counting} problem is the task of counting,
given a pattern $\tau \in \mathbb{S}_k$ and permutation $\pi \in \mathbb{S}_n$,
the number of occurrences $\pc{\tau}{\pi}$. Once again, there is a straightforward $\mathcal{O}(k \cdot n^k)$-time algorithm
-- how far is it from optimal?
Albert et al. lowered the bound to $\mathcal{O}(n^{2k/3 + 1})$ \cite{albert2001algorithms}\footnote{
Their algorithm also works for the counting version.
}
and the current best known bound is $\mathcal{O}(n^{(1/4 + o(1))k})$, due to Berendsohn et al. \cite{berendsohn2021finding}. 
Berendsohn et al. also showed a barrier: assuming the exponential time hypothesis,
there is no algorithm for pattern counting with running time $f(k) \cdot n^{o(k/\log k)}$,
for \textit{any} function $f$.

Another intriguing line of work focuses on the pattern counting problem, for \textit{constant small} $k$.
As the number of patterns $\tau \in \mathbb{S}_k$ is fixed in this regime,
one can equivalently, up to a constant multiplicative factor, compute the entire $k!$-dimensional vector of \textit{all} occurrences,
$(\pc{\tau}{\pi})_{\tau \in \mathbb{S}_k}$.
This vector, which characterises the local structure of a permutation over size-$k$ pointsets,
is known as the \textit{$k$-profile}. 
The $k$-profile has also featured in works 
aiming to understand the local structure of permutations,
c.f. \cite{beniaminilavee2023balanced, even2020patterns, cooper2008symmetric}.

Even-Zohar and Leng \cite{even2021counting} designed a class of algorithms capable of computing the $3$-profile
in $\Otilde{n}$-time,\footnote{
As usual, the notation $\Otilde{\cdot}$ hides poly-logarithmic factors.}
and the $4$-profile in $\Otilde{n^{3/2}}$-time.
Improving on their result for $k=4$, Dudek and Gawrychowski \cite{dudek2020counting} gave a bidirectional
reduction between
the  task of computing the $4$-profile, and that of counting $4$-cycles in a sparse graph.
The best known algorithm for the latter problem has running time $\mathcal{O}(n^{2-3/(2\omega + 1)})$ \cite{williams2014finding},
where $\omega < 2.372$ \cite{duan2023faster} is the exponent of matrix multiplication.
Consequently, Dudek and Gawrychowski obtain an $\mathcal{O}(n^{1.478})$-time algorithm 
for the $4$-profile. 
Our paper continues this line of work: we design algorithms computing the $5$, $6$ and $7$-profiles,
and highlight a barrier in the way of computing the $8$-profile.

\subsection{Our Contribution}

We introduce \emph{pattern-trees}: a family of graphs that generalise the corner-trees of Even-Zohar and Leng \cite{even2021counting}.
Pattern-trees are rooted labeled trees, in which every vertex is associated with a set of \textit{point variables},
along with constraints that fix their relative ordering in the plane,
and every edge is labeled by a list of constraints over the ordering
of points associated with its incident vertices. 

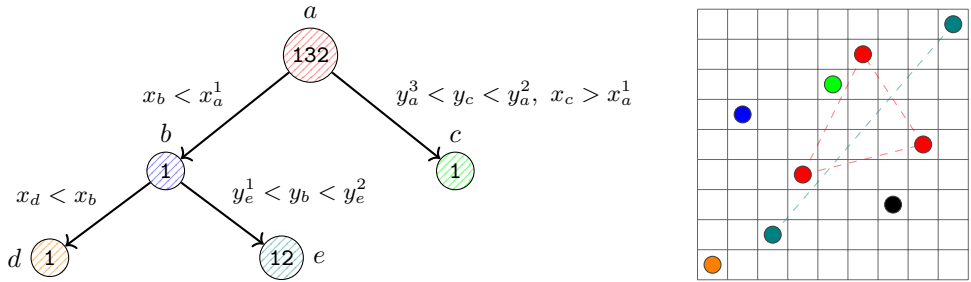
\begin{figure}[H]
    \centering
    \begin{tikzpicture}[scale=0.77]
	   \draw (0,-0.5) node[draw,circle, minimum size=0.5cm, inner sep=2pt, pattern=north east lines, distance=10pt, pattern color=red, fill opacity=0.4, text opacity=1,label=above:{$a$}] (root) {\small $\mathtt{132}$};
	   \draw (-2.5, -2.5) node[draw,circle, minimum size=0.5cm, inner sep=2pt, pattern=north east lines, distance=10pt, pattern color=blue, fill opacity=0.4, text opacity=1,label=above:{$b$}] (l) {\small $\mathtt{1}$};
	   \draw (2.5,-2.5) node[draw,circle, minimum size=0.5cm, inner sep=2pt, pattern=north east lines, distance=10pt, pattern color=green, fill opacity=0.5, text opacity=1,label=above:{$c$}] (r) {\small $\mathtt{1}$};
	   \draw (-4.5,-4) node[draw,circle, minimum size=0.5cm, inner sep=2pt, pattern=north east lines, distance=10pt, pattern color=orange, fill opacity=0.5, text opacity=1,label=left:{$d$}] (ll) {\small $\mathtt{1}$};
	   \draw (-0.5,-4) node[draw,circle, minimum size=0.5cm, inner sep=2pt, pattern=north east lines, distance=10pt, pattern color=teal, fill opacity=0.4, text opacity=1, label=right:{$e$}] (lr) {\small $\mathtt{12}$};
    
      \draw [thick, ->] (root) -- (l) node[midway,above left] {\small $x_b < x_a^1$};
      \draw [thick, ->] (root) -- (r) node[midway,above right] {\small $y_a^3 < y_c < y_a^2,\ x_c > x_a^1$};
      \draw [thick, ->] (l) -- (ll) node[midway,above left] {\small $x_d < x_b$};
      \draw [thick, ->] (l) -- (lr) node[midway,above right] {\small $y_e^1 < y_b < y_e^2$};
    \end{tikzpicture}
    \hspace{0.5cm}
    \begin{tikzpicture}[scale=0.4]
        \draw[help lines, color=darkgray, opacity=0.8] (0,0) grid (9,9);

        \draw [thin, dashed, draw=red, opacity=0.5] (3.5,3.5) -- (5.5,7.5);
        \draw [thin, dashed, draw=red, opacity=0.5] (3.5,3.5) -- (7.5,4.5);
        \draw [thin, dashed, draw=red, opacity=0.5] (7.5,4.5) -- (5.5,7.5);

        \draw [thin, dashed, draw=teal, opacity=0.5] (2.5,1.5) -- (8.5,8.5);

        \fill[red] (3.5,3.5) circle (8pt) {};
        \draw[darkgray] (3.5,3.5) circle (8pt) {};
        \fill[red] (5.5,7.5) circle (8pt) {};
        \draw[darkgray] (5.5,7.5) circle (8pt) {};
        \fill[red] (7.5,4.5) circle (8pt) {};
        \draw[darkgray] (7.5,4.5) circle (8pt) {};

        \fill[orange] (0.5,0.5) circle (8pt) {};
        \draw[darkgray] (0.5,0.5) circle (8pt) {};

        \fill[blue] (1.5,5.5) circle (8pt) {};
        \draw[darkgray] (1.5,5.5) circle (8pt) {};

        \fill[teal] (2.5,1.5) circle (8pt) {};
        \draw[darkgray] (2.5,1.5) circle (8pt) {};
        \fill[teal] (8.5,8.5) circle (8pt) {};
        \draw[darkgray] (8.5,8.5) circle (8pt) {};

        \fill[green] (4.5,6.5) circle (8pt) {};
        \draw[darkgray] (4.5,6.5) circle (8pt) {};

        \fill[black] (6.5,2.5) circle (8pt) {};
        \draw[darkgray] (6.5,2.5) circle (8pt) {};
        
    \end{tikzpicture}
    \caption{An embedding of a pattern-tree (left) into the permutation $\mathtt{162478359} \in \mathbb{S}_9$ (right).}
    \label{fig:pattern_tree_simple}
\end{figure}

Using an algorithm derived from pattern-trees, we obtain our first result.
\begin{introthm}
    \label{introthm:fast_le_7_profile}
    For every $1 \le k \le 7$, the $k$-profile of an $n$-element permutation can be computed in~$\Otilde{n^2}$ time and space.
\end{introthm}

Our proof of \Cref{introthm:fast_le_7_profile} relies on embeddings of trees into permutations.
Consider the \textit{number of distinct embeddings} of the points of a pattern-tree $T$
into the points in the plane associated with a permutation $\pi \in \Sn$,
in which the embedding satisfies all constraints defined by the tree. 
We show that this quantity
can be expressed as a \textit{fixed} integer linear combination
of permutation pattern counts, irrespective of $\pi$.
Interpreted as a formal sum of patterns,
this is simply a vector in $\ZZ^{\mathbb{S}_{\le k}}$,
where $k$ is the number of point variables in the tree.
These vectors are associated
with pairwise compositions of linear extensions of partially ordered sets,
whose Hasse diagrams can be partitioned in a particular way.

The subspaces spanned by the vectors of trees, over the rationals, are central to our proof. 
It is not hard to show that when the subspace of a set of trees is full-dimensional,
one can derive from those trees an algorithm for the $k$-profile.
To this end, we design an \textit{evaluation algorithm}: 
given a pattern-tree $T$ and an input permutation $\pi \in \Sn$\footnote{
Throughout this paper we operate on $n$-element permutations as input. Such inputs are assumed to be presented to the algorithm \textit{sparsely},  e.g., as a length-$n$ vector representing the permutation in one-line notation.},
the algorithm computes the number of occurrences of $T$ in $\pi$, denoted $\#T(\pi)$.
The complexity of this algorithm depends on properties of the tree.
In our proof of \Cref{introthm:fast_le_7_profile},
we construct a family of trees evaluable in $\Otilde{n^2}$-time,
which are of full dimension for $\mathbb{S}_{\le 7}$. 

Compared to previous results, \Cref{thm:fast_le_7_profile} offers an improvement whenever $k \in \{5,6,7\}$. 
The best known bound for the $k$-profile problem is $\mathcal{O}(n^{k/4 + o(k)})$, due to Berendsohn et al. \cite{berendsohn2021finding}.
Their approach relies on formulating a binary CSP, and bounding its tree-width.
It is well known that binary CSPs can be solved in time $\mathcal{O}(n^{t+1})$ \cite{dechter1989tree, freuderl1990complexity},
where $n$ is the domain size, and $t$ is the tree-width of the constraint graph.
In the algorithm of \cite{berendsohn2021finding}, the tree-width is bounded by $k/4 + o(k)$,
where the $o(k)$-term is greater than one.
Therefore, their algorithm has at least cubic running time when $k \ge 4$. \ \\

The relationship between properties of pattern-trees and the dimensions of the subspaces spanned by them is still far from understood (see \Cref{sect:discussion}).
Corner-trees,
which are exactly the pattern-trees whose evaluation is quasi-linear,
were shown in \cite{even2021counting} to have full rank for $\mathbb{S}_{\le 3}$,
and rank only $|\mathbb{S}_4|-1 = 23$, restricted to $\mathbb{S}_4$. 
Intriguingly, we show that the family of pattern-trees with which our proof of \Cref{thm:fast_le_7_profile} is obtained,
whose evaluation complexity is quadratic, have full rank for $\mathbb{S}_{\le 7}$,
and rank only $|\mathbb{S}_8|-1 = 40319$ restricted to $\mathbb{S}_8$.
We observe several striking resemblances between the two vectors spanning the orthogonal complements,
for $\mathbb{S}_4$ and $\mathbb{S}_8$ respectively, in terms of their symmetries.
In fact, we extend a characterisation of \cite{dudek2020counting} regarding the symmetries for $\mathbb{S}_4$
to the case of $\mathbb{S}_8$ (see \Cref{subsect:s_8_not_spanned}). \ \\

Our second result is a sub-quadratic algorithm for the $5$-profile.
\begin{introthm}
    \label{introthm:fast-5-prof}
    The $5$-profile of an $n$-element permutation can be computed in time $\Otilde{n^{7/4}}$.
\end{introthm}

The proof of \Cref{introthm:fast-5-prof} is obtained by speeding-up the evaluation algorithm of pattern-trees.
The original algorithm for pattern-trees 
has an integral exponent in its complexity, which is determined by properties of the tree.
We show that trees with certain topological properties, i.e., containing a particular set of ``gadgets'', can be evaluated faster.
The family of trees constituting all corner-trees, and their augmentation by our gadgets,
span a full-dimensional subspace over $\mathbb{S}_5$.
This allows us to break the quadratic barrier for the $5$-profile.

One of the key ingredients, both in the original evaluation algorithm and in its extended version,
is a data structure known as a multidimensional segment-tree, or rectangle-tree \cite{jaja2005space,chazelle1988functional}.\footnote{A $2$-dimensional version of this data structure features in both \cite{even2021counting} and \cite{dudek2020counting}.}
A $d$-dimensional rectangle-tree holds (possibly weighted) points in $[n]^d$,
and answers sum-queries over rectangles $\mathcal{R} \subseteq [n]^d$ (i.e., Cartesian products of segments)
in poly-logarithmic time. 

The gadgets appearing in the proof of \Cref{introthm:fast-5-prof} are sub-structures related
to the patterns $\mathtt{3214}$ and $\mathtt{43215}$.
For the former, we extend an algorithm of \cite{even2021counting} into a weighted variant,
and provide an evaluation algorithm  of complexity $\widetilde{\mathcal{O}}(n^{5/3})$.
We then further extend this into an algorithm for the latter gadget, of complexity $\widetilde{\mathcal{O}}(n^{7/4})$.
The latter proof is involved, and requires the introduction of a new data structure, which we call
a \textit{pair-rectangle-tree}.
A pair-rectangle-tree is an extension of rectangle-trees that can facilitate more complex queries,
in particular, regarding the dominance counting (see \cite{jaja2005space, chazelle1987linear}) of a set of points in a rectangle.
We remark that the original pattern-tree evaluation algorithm 
can only compute \textit{equivalent} gadgets in quadratic time.
That is, the evaluation algorithm is not always optimal.

\subsection{Paper Organization}

In \Cref{sect:pattern_trees} we introduce pattern-trees.
Our construction for $3 \le k \le 7$ can be found in \Cref{subsect:5_to_7_prof},
and the case $k=8$ is dealt with in \Cref{subsect:s_8_not_spanned}.
A straightforward application of pattern-trees for general $k$ is given in \Cref{subsect:k_over_2_alg}. 
\Cref{sect:5_prof_alg} revolves around our construction of an $\Otilde{n^{7/4}}$-time
algorithm for the $5$-profile.
The augmentation of the pattern-trees evaluation algorithm can be found in \Cref{subsect:bottom-up-improve},
and the particular gadgets used in the $5$-profile are obtained in \Cref{subsect:3214_gadget} and \Cref{subsect:43215_gadget}.
The data structure we introduce for dominance counting in rectangles, pair-rectangle-tree, is given in \Cref{subsect:pair_rect_trees}.
Finally, in \Cref{sect:discussion} we discuss open questions and possible extensions of this work.

\section{Preliminaries}
\label{sect:prelims}
\subsection{Permutations}

A permutation $\pi \in \mathbb{S}_n$ over $n$ elements is a bijection from $[n]$ to itself, where $[n] \eqdef \{1, 2, \dots, n\}$. Throughout this paper, we express permutations using one-line notation, and if the permutation range is sufficiently small, we omit the parentheses. For instance, $\mathtt{123}$ is the identity permutation over $3$ elements. Associated with any permutation $\pi \in \Sn$ is a set of $n$ points in the plane, $p(\pi) \eqdef \{ (i, \pi(i)) : i \in [n]\}$, which we refer to as the \emph{points of $\pi$}. In the other direction, any set of $n$ points in the plane defines a permutation $\pi \in \Sn$, provided that no two points lie on an axis-parallel line. Given such a set $S \subset \mathbb{R}^2$, we use the notation $S \cong \pi$ to indicate that the points are \emph{order-isomorphic} to $\pi$.

An \emph{occurrence} of a \emph{pattern} $\tau \in \mathbb{S}_k$ in a permutation $\pi \in \Sn$ is a $k$-tuple  $1\le i_1 < \cdots< i_k \le n$ such that the set of points $(i_j,\pi(i_j))$ is order-isomorphic to $\tau$. That is, $\pi(i_j) < \pi(i_l)$ if and only if $\tau(j)<\tau(l)$ for all $j,l\in[k]$. The number of occurrences of $\tau$ in $\pi$ is denoted by $\pc{\tau}{\pi}$.

The dihedral group $D_4$ naturally acts on the symmetric group $\Sn$,
by acting on $[1, n]^2$. Formally, for any element $g \in D_4$ and 
permutation $\pi \in \Sn$, we have that $(g. \pi) \in \Sn$ is the permutation 
for which $g. (p(\pi)) \cong g. \pi$. 
Our algorithms usually receive permutations as input, and compute some combination of pattern-counts.
To this end, it is sometimes helpful to first act on the input with an element $g \in D_4$  (as a preprocessing step),
and only then invoke the algorithm as usual.
In this way, if an algorithm computes the count $\pc{\tau}{\pi}$, then after the action we obtain
$\pcn{\tau}{g. \pi} = \pcn{(g^{-1}. \tau)}{\pi}$.

Our main focus in this paper is the computation of $\pc{\tau}{\pi}$ for all $\tau\in\mathbb{S}_k$, where $\pi\in\Sn$ is given as input and and $k$ is fixed. This collection of counts is defined as follows.

\begin{definition}
    The \emph{$k$-profile} of a permutation $\pi \in \Sn$ is the vector
    $\left( \pc{\tau}{\pi} \right)_{\tau \in \mathbb{S}_k} \in \ZZ^{\mathbb{S}_k}$.
\end{definition}

\subsection{Partially Ordered Sets}

A partially ordered set (\emph{poset}) $\mathcal{P}(X, \le)$ over a ground set $X$
is a partial arrangement of the elements in $X$ according to the order relation $\le$.
If $\le$ is not reflexive, we say that $\mathcal{P}$ is \emph{strict}.
A partial order $\le^\star$ is said to be an extension
of $\le$ if $x\le y$ implies $x\le^\star y$ for all $x,y\in X$. If an extension $\le^\star$ is a total order, it is called a \emph{linear extension} of $\le$. As usual, the set of all linear extensions of a poset $\mathcal{P}$
is denoted by $\mathcal{L}(\mathcal{P})$.

\subsection{Computational Model}

Throughout this paper we disregard all $\polylog(n)$-factors, so our results hold for any choice of standard computational model (say, word-RAM). The notation $\Otilde{n^k}$ (adding the tilde) is used to hide poly-logarithmic factors. The algorithms presented in this paper operate on $n$-element permutations as input, and we remark that such inputs are assumed to be presented to the algorithm \textit{sparsely}, e.g., as a length-$n$ vector representing the permutation in one-line notation. 

\subsection{Rectangle-Trees}
\label{subsect:rect_tree}
Our algorithms for efficiently computing profiles rely heavily on a simple and powerful data structure, which we refer to as a \emph{rectangle-tree}\footnote{
    A rectangle $\mathcal{R} \subseteq [n]^d$ is a Cartesian product of \textit{segments}, i.e., invervals of the form $\{a, a+1, \dots, b\} \subseteq [n]$.  
} or a \emph{multidimensional segment-tree}. Concretely, we require the following folklore fact. 

\begin{proposition}[\cite{chazelle1988functional, jaja2005space}, see also \cite{dudek2020counting}]
    \label{prop:rect_trees}
    For any fixed dimension $d\ge 1$, there exists a deterministic data structure $\mathcal{T}$ that supports each of the following actions in $\Otilde{1}$ time:
    \begin{enumerate}
        \item  \underline{Initialisation}: Given $n \in \mathbb{N}$, construct an empty tree over $[n]^d$.
        \item \underline{Insertion}: Given $x\in [n]^d$ and $w = \mathcal{O}\left(\poly(n)\right)$, add weight $w$ to point $x$.
        \item \underline{Query}: Given a rectangle $\mathcal{R} \subseteq [n]^d$, the query $\mathcal{T}(\mathcal{R})$ returns the sum of weights over all points in $\mathcal{R}$.
    \end{enumerate}
\end{proposition}
Let us illustrate the application of rectangle-trees to pattern counting, through the simple (and again, folklore) case of \textit{monotone pattern counting}.
\begin{proposition}
    \label{prop:count-123k}
    Let $k\geq 1$ be a fixed integer and let $\pi\in\Sn$ be an input permutation.
    The pattern counts $\pc{(1,\ldots,k)}{\pi}$ and $\pc{(k,\ldots,1)}{\pi}$ can be computed in $\Otilde{n}$ time.
\end{proposition}
\begin{proof}
    Without loss of generality, we count the ascending pattern.
    Construct a $2$-dimensional rectangle-tree $\mathcal{T}_1$, and insert every point $(i,\pi(i))\in p(\pi)$ with weight $1$.
    Note that the rectangle query $\mathcal{T}_1([1,i-1]\times [1,\pi(i)-1])$ counts how many occurrences of $\mathtt{12}$ end in $(i,\pi(i))$.
    For every $i$, we use that value as the weight of $(i,\pi(i))$ in a new $2$-dimensional rectangle-tree, $\mathcal{T}_2$.
    
    Continuing inductively, for every $2\leq j \leq k$,
    the point $(i,\pi(i))$ is inserted into a $2$-dimensional tree $\mathcal{T}_j$ with weight $\mathcal{T}_{j-1}([1,i-1]\times [1,\pi(i)-1])$. This counts occurrences of  $\mathtt{(1,\ldots,j)}$ ending in $(i,\pi(i))$.
    The final answer is given by $\mathcal{T}_k([n]\times [n])$.
    The complexity is $\Otilde{kn}=\Otilde{n}$,
    since for each of the $n$ permutation points and each of the $k$ trees we perform one query and one insertion.
\end{proof}
\begin{rem}
    It is also possible to count monotone patterns using $1$-dimensional segment-trees, somewhat more efficiently.
    However, the difference is only in logarithmic factors.
    The multidimensional structure highlighted above will serve us in more complicated cases.
\end{rem}
\section{Pattern-Trees}
\label{sect:pattern_trees}

In this section we introduce a family of graphs, called \textit{pattern-trees}.
Using pattern-trees we derive algorithms for computing the $k$-profile of a permutation.
Our main result for this section (see \Cref{subsect:5_to_7_prof}) is
a quadratic-time algorithm for the $k$-profile of a permutation, for every $k \le 7$:
\begin{thm}
    \label{thm:fast_le_7_profile}
    For $1 \le k \le 7$, the $k$-profile of an $n$-element permutation is computable in $\Otilde{n^2}$ time and space.
\end{thm}

In \Cref{subsect:s_8_not_spanned} we consider the subspaces spanned by the same family of pattern-trees, restricted to $\mathbb{S}_8$.
We show that this subspace is of dimension $|\mathbb{S}_8| - 1$, one less than required.
In \Cref{subsect:k_over_2_alg} we consider the case of general (constant) $k$, and show a straightforward 
application of pattern-trees yielding an $\widetilde{\mathcal{O}}(n^{\lceil k/2 \rceil})$-time algorithm
for the $k$-profile.\ \\

Before we present pattern-trees, let us begin by recalling corner-trees.

\subsection{Warmup: Corner-Trees}

One of the main components in the work of \cite{even2021counting} is the introduction of \emph{corner-trees}.
Corner-trees are a family of rooted edge-labeled trees. 
Every corner-tree of $k$ vertices is associated with a particular vector in $\ZZ^{\mathbb{S}_{\le k}}$;
i.e., a formal integer linear combination of permutations, each of size at most $k$.
Furthermore, there exists an efficient evaluation algorithm for corner-trees:
given any input permutation $\pi \in \Sn$ and corner-tree $T$,
the integer sum of permutation \textit{pattern counts} in $\pi$, called the vector of $T$,
can be computed in time $\Otilde{n}$. We refer to this operation
as \emph{evaluating the vector of $T$ over $\pi$}.

\begin{definition}[corner-tree \cite{even2021counting}\footnote{For convenience, we consider corner-trees to be edge-labeled, rather than vertex-labeled as in \cite{even2021counting}.}]
A corner-tree is a rooted\footnote{Hereafter, whenever we consider rooted trees, we orient their edges away from the root.} edge-labeled tree, with edge labels in the set $\{\mathrm{NE}, \mathrm{NW}, \mathrm{SE}, \mathrm{SW}\}$.
\end{definition}

An \textit{occurrence} of a corner-tree $T$ in a permutation $\pi$ is a map $\varphi: V(T) \to p(\pi)$, in which the image \textit{agrees} with the edge-labels of the tree. That is, for every edge $(u \to v) \in E(T)$, $\varphi(v)$ is to the \textit{left} of $\varphi(u)$ if the edge is labeled $\text{NW}$ or $\text{SW}$, and to its right otherwise. Similar rules apply for their vertical ordering.
As in \cite{even2021counting}, the \textit{number of occurrences} of a corner-tree $T$ in a permutation $\pi$ is denoted by $\#T(\pi)$.

The \textit{vector} of a corner-tree is a formal sum of permutation patterns with integer coefficients, representing the number of occurrences of the tree in \textit{any} input permutation.
For instance, the vector of \begin{tikzpicture}[scale=0.2, every node/.style={inner sep=0,outer sep=0}]
	   \draw (0,0) node[draw,circle, scale=0.3, minimum size=5mm, fill=black] (root) {};
	   \draw (4,0) node[draw,circle, scale=0.3, minimum size=5mm] (c1) {};
	   \draw (8,0) node[draw,circle, scale=0.3, minimum size=5mm] (c2) {};    
      \draw [thin, <-, scale=0.1] (c1) -- (root) node[midway,above=0.3mm] {\tiny SE};
      \draw [thin, <-, scale=0.1] (c2) -- (c1) node[midway,above=0.3mm] {\tiny NE};
\end{tikzpicture} is $\pce{213} + \pce{312}$. Clearly, the vector of a corner-tree over $k$ vertices may involve patterns of size at most $k$, as the tree conditions on the relative ordering of at most $|V(T)|$ points (smaller patterns may appear as well, since occurrences are not necessarily injective).

\begin{figure}[H]
    \centering
    \begin{tikzpicture}[scale=0.75]
	   \draw (0,-0.5) node[draw,circle, minimum size=0.5cm, inner sep=2pt, pattern=north east lines, distance=10pt, pattern color=red] (root) {};
	   \draw (-2, -2) node[draw,circle, minimum size=0.5cm, inner sep=2pt, pattern=north east lines, distance=10pt, pattern color=blue] (l) {};
	   \draw (2,-2) node[draw,circle, minimum size=0.5cm, inner sep=2pt, pattern=north east lines, distance=10pt, pattern color=green] (r) {};
	   \draw (-3.5,-3) node[draw,circle, minimum size=0.5cm, inner sep=2pt, pattern=north east lines, distance=10pt, pattern color=orange] (ll) {};
	   \draw (-0.5,-3) node[draw,circle, minimum size=0.5cm, inner sep=2pt, pattern=north east lines, distance=10pt, pattern color=purple] (lr) {};
    
      \draw [thick, ->] (root) -- (l) node[midway,above left] {\small NW};
      \draw [thick, ->] (root) -- (r) node[midway,above right] {\small NW};
      \draw [thick, ->] (l) -- (ll) node[midway,above left] {\small SE};
      \draw [thick, ->] (l) -- (lr) node[midway,above right] {\small SW};
    \end{tikzpicture}
    \hspace{1cm}
    \begin{tikzpicture}[scale=0.4]
        \draw[help lines, color=darkgray, opacity=0.8] (0,0) grid (7,7);

        \draw [thin, -latex, ->] (3.5,0.5) -- (2.49,6.15);
        \draw [thin, -latex, ->] (3.5,0.5) -- (1.55,3.2);
        \draw [thin, -latex, ->] (1.5,3.5) -- (5.17,2.55);
        \draw [thin, -latex, ->] (1.5,3.5) -- (0.6,1.8);

        \fill[red] (3.5,0.5) circle (8pt) {};
        \draw[darkgray] (3.5,0.5) circle (8pt) {};
        \fill[green] (2.5,6.5) circle (8pt) {};
        \draw[darkgray] (2.5,6.5) circle (8pt) {};
        \fill[blue] (1.5,3.5) circle (8pt) {};
        \draw[darkgray] (1.5,3.5) circle (8pt) {};
        \fill[orange] (5.5,2.5) circle (8pt) {};
        \draw[darkgray] (5.5,2.5) circle (8pt) {};
        \fill[purple] (0.5,1.5) circle (8pt) {};
        \draw[darkgray] (0.5,1.5) circle (8pt) {};
        \fill[black] (4.5,5.5) circle (8pt) {};
        \draw[darkgray] (4.5,5.5) circle (8pt) {};
        \fill[black] (6.5,4.5) circle (8pt) {};
        \draw[darkgray] (6.5,4.5) circle (8pt) {};        
        
    \end{tikzpicture}
    \hspace{1cm}
    \begin{tikzpicture}[scale=0.4, fill fraction/.style n args={2}{path picture={
 \fill[#1] (path picture bounding box.south west) rectangle
 ($(path picture bounding box.north west)!#2!(path picture bounding box.north
 east)$);}}]
        
        \draw[help lines, color=darkgray, opacity=0.8] (0,0) grid (7,7);
        
        \draw [thin, -latex, ->] (3.5,0.5) -- (1.55,3.2);
        \draw [thin, -latex, ->] (1.5,3.5) -- (5.17,2.55);
        \draw [thin, -latex, ->] (1.5,3.5) -- (0.6,1.8);

        \fill[red] (3.5,0.5) circle (8pt) {};
        \draw[darkgray] (3.5,0.5) circle (8pt) {};
        \fill[darkgray] (2.5,6.5) circle (8pt) {};
        \draw[darkgray] (2.5,6.5) circle (8pt) {};
        \fill[blue] (1.5,3.5) circle (8pt) {};
        \draw (1.5,3.5) node[circle, draw, fill fraction={green}{0.5}, inner sep=0pt,minimum size=6pt] (8pt) {};
        \draw[darkgray] (1.5,3.5) circle (8pt) {};
        \fill[orange] (5.5,2.5) circle (8pt) {};
        \draw[darkgray] (5.5,2.5) circle (8pt) {};
        \fill[purple] (0.5,1.5) circle (8pt) {};
        \draw[darkgray] (0.5,1.5) circle (8pt) {};
        \fill[black] (4.5,5.5) circle (8pt) {};
        \draw[darkgray] (4.5,5.5) circle (8pt) {};
        \fill[black] (6.5,4.5) circle (8pt) {};
        \draw[darkgray] (6.5,4.5) circle (8pt) {};        
        
    \end{tikzpicture}

    \caption{Two occurrences of a corner-tree (left) in $\pi = \mathtt{2471635} \in \mathbb{S}_7$ (centre, right).  Occurrences need not be injective; for instance, on the right, the blue and green points are identified.}
    \label{fig:corner_tree_occurrence}
\end{figure}
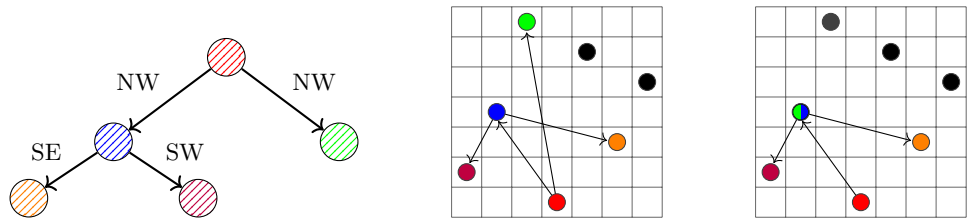

Theorem 1.1 of \cite{even2021counting} presents an algorithm for evaluating the vector of a corner-tree over an input permutation $\pi \in \Sn$. 
For expositionary purposes, we sketch a simplified version of their algorithm, phrased in terms of rectangle-trees.

\begin{proposition} [Theorem 1.1 of \cite{even2021counting}]
    \label{prop:compute_corner_tree}
    The vector of any corner-tree with a constant number of vertices can be evaluated over an input permutation $\pi \in \Sn$ in time $\Otilde{n}$.
\end{proposition}
\begin{proofsketch}
    Let $T$ be a corner-tree and let $\pi \in \Sn$ be a permutation. To start, construct a $2$-dimensional rectangle-tree (see \Cref{subsect:rect_tree}), and insert the points $p(\pi)$ with weight $1$, in time $\Otilde{n}$. Associate this tree with the \textit{leaves} of $T$. Next, traverse the vertices of $T$ in post-order. At every internal vertex $u$, construct a new (empty) rectangle-tree $\mathcal{T}_u$, and associate it with $u$. Then, iterate over every point in $\pi$, and at each point perform one \textit{rectangle query} to the rectangle-tree associated with each of $u$'s children, querying the rectangle corresponding to the edge label in $T$ written on the parent-child edge. For example, if $u \to v$ is labeled $\text{SW}$, the iteration over a point $(i,\pi(i))\in p(\pi)$ queries the rectangle $[1,i-1]\times [1,\pi(i)-1]$. Store the product of all answers to these queries in $\mathcal{T}_u$, at the position of the current permutation point. It can be shown that the sum of all values at the root's tree (i.e., a full rectangle query) is the number of occurrences, $\#T(\pi)$.   
\end{proofsketch}

\subsection{Pattern-Trees}

We introduce \emph{pattern-trees}: a family of graphs that generalise the corner-trees of \cite{even2021counting}.
In pattern-trees, every \textit{vertex} is labeled by a permutation, and every \textit{edge} 
is labeled by a list of constraints. 
The permutations written on the vertices fix the exact ordering of the points corresponding to them,
and the edge-constraints are similarly imposed over the points corresponding to the two incident vertices.  
As with corner-trees, pattern-trees serve two purposes:
firstly, every pattern-tree is associated with a set of constraints over permutation points, the
number of satisfying assignments to which can be expressed as a formal integer linear combination of patterns (that is, a vector).
Secondly, we present an algorithm for evaluating this vector over an input permutation.
This allows us to efficiently compute certain pattern combinations not spanned by corner-trees.
\begin{definition}[pattern-tree]
    \label{defn:pattern_tree}
    A pattern-tree $T$ is a rooted edge- and vertex-labeled tree, where:
    \begin{enumerate}
        \item Every vertex $v \in V(T)$ is:
        \begin{itemize}
            \item Labeled by a permutation $\tau_v \in \mathbb{S}_r$, for some integer $r \ge 1$.
            \item Associated with two sets of \textit{fresh} variables, 
            \[ x_v \eqdef \{ x_v^1, \dots, x_v^r \}, \text{ and } y_v \eqdef \{ y_v^1, \dots, y_v^r\}, \]
            where we denote $p_v^i \eqdef (x_v^i, y_v^i)$ for every $i \in [r]$, and $p_v \eqdef \{ p_v^i : i \in [r] \}$.
        \end{itemize}
        \item Every edge $(u \to v) \in E(T)$ is labeled by:
        \begin{itemize}
            \item Two strict posets, $\mathcal{P}^x_{uv} = (x_u \sqcup x_v, <)$ and $\mathcal{P}^y_{uv}=(y_u \sqcup y_v, <)$.
            \item A set $E_{uv} \subseteq p_u \times p_v$ of equalities between the points of $u$ and those of $v$.
        \end{itemize}
    \end{enumerate}
\end{definition}
The \emph{size} $s(v)$ of a vertex $v$ is the size $r$ of the permutation $\tau_v \in \mathbb{S}_r$ with which it is labeled. The \emph{maximum size} of a pattern-tree, denoted $s(T)$, is the maximum over all vertex sizes. The \emph{total size}, denoted $\Sigma(T)$, is the sum over all vertex sizes. Under this notation, a corner-tree is a pattern-tree of maximum size one. Lastly, $p(T) \eqdef \bigsqcup_{v \in V(T)} p_v$ is the set of all $\Sigma(T)$ points in the tree.

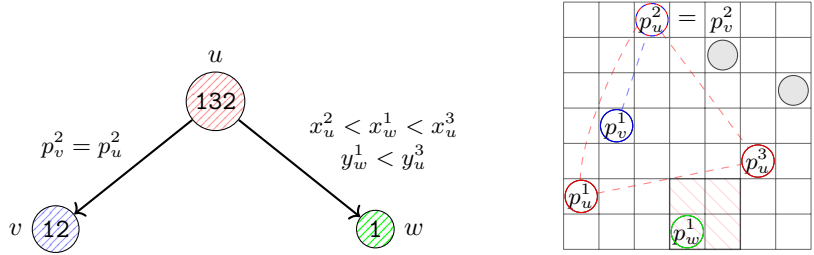
\begin{figure}[H]
    \centering
    \begin{tikzpicture}[scale=0.85]
	   \draw (0,-0.5) node[draw,circle, minimum size=0.5cm, inner sep=2pt, pattern=north east lines, distance=10pt, pattern color=red, fill opacity=0.4, text opacity=1, label=above:{$u$}] (root) {$\mathtt{132}$};
	   \draw (-2.5, -2.5) node[draw,circle, minimum size=0.5cm, inner sep=2pt, pattern=north east lines, distance=10pt, pattern color=blue, fill opacity=0.4, text opacity=1, label=left:{$v$}] (l) {$\mathtt{12}$};
	   \draw (2.5,-2.5) node[draw,circle, minimum size=0.5cm, inner sep=2pt, pattern=north east lines, distance=10pt, pattern color=green, label=right:{$w$}] (r) {$\mathtt{1}$};

      \draw [thick, ->] (root) -- (l) node[midway,above left] {\small $p_v^2 = p_u^2$};
      \draw [thick, ->] (root) -- (r) node[midway,above right=-0.15cm and 0.0cm] {\small \shortstack{$x_u^2 < x_w^1 < x_u^3$ \\ $y_w^1 < y_u^3$}};
    \end{tikzpicture}
    \hspace{1cm}
    \begin{tikzpicture}[scale=0.47]
    
        \tikzset{
          rbcirc/.style={circle,fill=white,draw=blue,dash pattern= on 2pt off 2pt,postaction={draw,red,dash pattern= on 2pt off 2pt,dash phase=2pt},text width=1.5em}
        }
        
        \draw[help lines, color=darkgray, opacity=0.8] (0,0) grid (7,7);

        \node[fit={(3,0) (5,2)}, inner sep=0pt, line width=0.2mm, pattern=south east lines, draw=black, pattern color=red, opacity=0.4] (rect) {};
  
        \draw[black] (3.5,0.5) circle (13pt) node {\small $p_w^1$};
        \draw[green] (3.5,0.5) circle (13pt) {};

        \draw [thin, dashed, draw=red, opacity=0.5] (2.5,6.5) -- (5.5,2.5);
        \draw [thin, dashed, draw=red, opacity=0.5] (0.5,1.5) to[out=100,in=70] (2.5,6.5);
        \draw [thin, dashed, draw=red, opacity=0.5] (0.5,1.5) -- (5.5,2.5);
        \draw [thin, dashed, draw=blue, opacity=0.5] (1.5,3.5) -- (2.5,6.5);

        \draw[rbcirc] (2.5,6.5) circle (13pt) node {};
        \draw[fill=none] (2.5,6.5) node {\small $p_u^2$};
        \draw[fill=none] (3.5, 6.5) node {\small $=$};
        \draw[fill=none] (4.5, 6.5) node {\small $p_v^2$};
        
        \draw[black, fill=white] (1.5,3.5) circle (13pt) node {\small $p_v^1$};
        \draw[blue] (1.5,3.5) circle (13pt) {};

        \draw[black, fill=white] (5.5,2.5) circle (13pt) node {\small $p_u^3$};
        \draw[red] (5.5,2.5) circle (13pt) {};
        
        \draw[black, fill=white] (0.5,1.5) circle (13pt) node {\small $p_u^1$};
        \draw[red] (0.5,1.5) circle (13pt) {};

        \fill[black, opacity=0.1] (4.5,5.5) circle (12pt) {};
        \draw[darkgray] (4.5,5.5) circle (12pt) {};
        
        \fill[black, opacity=0.1] (6.5,4.5) circle (12pt) {};
        \draw[darkgray] (6.5,4.5) circle (12pt) {};            
    \end{tikzpicture}
    \caption{An occurrence of a pattern-tree $T$ (left) in the permutation $\pi = \mathtt{2471635} \in \mathbb{S}_7$ (right).
    Every set of coloured points on the right induces the permutation with which the similarly coloured vertex on the left is labeled (``vertex constraints'').
    All of the edge-constraints are also satisfied: points $p_v^2$ and $p_u^2$ are \textit{identified},
    and point $p_w^1$ (green) must reside within the red shaded square.
    This tree corresponds to a linear combination, $\pce{1423} + \pce{2413} + 2 \cdot \pce{12534} + \dots + \pce{24513}$, of patterns in $\mathbb{S}_4$ and $\mathbb{S}_5$.
    The tree has total size $\Sigma(T)=6$ and maximum size $s(T)=3$.
   }
    \label{fig:pattern_tree}
\end{figure}


\paragraph{Pattern-Tree Constraints.} Any pattern-tree $T$ defines constraints $\mathcal{C}(T)$ over \textit{points} $p(T)$:
\begin{enumerate}
    \item Every vertex $v$ labeled by $\tau_v \in \mathbb{S}_r$ contributes the following inequalities,\footnote{These vertex-constraints  enforce the pattern $\tau_v$ over the points $p_v$.} 
    \[
        x_v^1 < x_v^2 < \cdots < x_v^r, \text{ and } y_v^i < y_v^j \text { for all $i,j\in [r]$ such that $\tau_v(i)<\tau_v(j)$}.
    \]
    \item Every edge $u \to v$ contributes the inequalities in $\mathcal{P}_{uv}^x$ and $\mathcal{P}_{uv}^y$, and the equalities in $E_{uv}$.
\end{enumerate}
Hereafter, we partition $\mathcal{C}(T)$ into two parts: its \textit{equalities},
which define an equivalence relation $E^T \eqdef \bigsqcup_{u \to v} E_{uv}$ over the points $p(T)$,
and its \textit{inequalities}, which define \textit{strict} posets,
\[
    \mathcal{P}_x^T = \Big(\bigsqcup_{v \in V(T)} x_v, < \Big), \text{ and }\mathcal{P}_y^T = \Big(\bigsqcup_{v \in V(T)} y_v, < \Big).
\]

Given an equivalence relation $E \supseteq E^T$, the posets $\mathcal{P}_x^{E}$ and $\mathcal{P}_y^{E}$
are the strict posets obtained from $\mathcal{P}_x^T$ and $\mathcal{P}_y^T$ by replacing every
coordinate variable corresponding to a point $p \in p(T)$ by a single variable corresponding to the equivalence class of $p$ in $E$.

\begin{exmp}
The pattern-tree $T$ appearing in \Cref{fig:pattern_tree} corresponds to the constraints
\[
    \mathcal{C}(T) = \left\{ x_u^1 < x_u^2 < x_u^3,\ y_u^1 < y_u^3 < y_u^2,\ x_u^2 < x_w^1 < x_u^3,\ y_w^1 < y_u^3,\ p_v^2 = p_u^2,\ x_v^1 < x_v^2,\ y_v^1 < y_v^2 \right\}, 
\]
whose posets are:
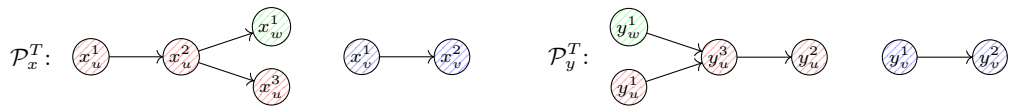
\begin{figure}[H]
    \centering
    \begin{tikzpicture}[scale=0.4, every node/.style={inner sep=0,outer sep=0}]
        \draw[fill=none] (-2, 0) node {\small $\mathcal{P}_x^T$:};
        \draw (0,0) node[draw,circle, minimum size=0.1cm, inner sep=0.5pt, text opacity=1, pattern=north east lines, distance=10pt, pattern color=red, fill opacity=0.3] (x_u_1) {\scriptsize $x_u^1$};
        \draw (3,0) node[draw,circle, minimum size=0.1cm, inner sep=0.5pt, text opacity=1, pattern=north east lines, distance=10pt, pattern color=red, fill opacity=0.3] (x_u_2) {\scriptsize $x_u^2$};
        \draw (6,-1) node[draw,circle, minimum size=0.1cm, inner sep=0.5pt, text opacity=1, pattern=north east lines, distance=10pt, pattern color=red, fill opacity=0.3] (x_u_3) {\scriptsize $x_u^3$};
        \draw (6,1) node[draw,circle, minimum size=0.1cm, inner sep=0.5pt, text opacity=1, pattern=north east lines, distance=10pt, pattern color=green, fill opacity=0.3] (x_w_1) {\scriptsize $x_w^1$};
        \draw (9,0) node[draw,circle, minimum size=0.1cm, inner sep=0.5pt, text opacity=1, pattern=north east lines, distance=10pt, pattern color=blue, fill opacity=0.3] (x_v_1) {\scriptsize $x_v^1$};
        \draw (12,0) node[draw,circle, minimum size=0.1cm, inner sep=0.5pt, text opacity=1, pattern=north east lines, distance=10pt, pattern color=blue, fill opacity=0.3] (x_v_2) {\scriptsize $x_v^2$};
        \draw [->] (x_u_1) -- (x_u_2) node[midway,above left] {};
        \draw [->] (x_u_2) -- (x_u_3) node[midway,above left] {};
        \draw [->] (x_u_2) -- (x_w_1) node[midway,above left] {};
        \draw [->] (x_v_1) -- (x_v_2) node[midway,above left] {};
    \end{tikzpicture}
    \hspace{0.8cm}
    \begin{tikzpicture}[scale=0.4, every node/.style={inner sep=0,outer sep=0}]
        \draw[fill=none] (-2, 0) node {\small $\mathcal{P}_y^T$:};
        \draw (0,-1) node[draw,circle, minimum size=0.1cm, inner sep=0.5pt, text opacity=1, pattern=north east lines, distance=10pt, pattern color=red, fill opacity=0.3] (y_u_1) {\scriptsize $y_u^1$};
        \draw (6,0) node[draw,circle, minimum size=0.1cm, inner sep=0.5pt, text opacity=1, pattern=north east lines, distance=10pt, pattern color=red, fill opacity=0.3] (y_u_2) {\scriptsize $y_u^2$};
        \draw (3,0) node[draw,circle, minimum size=0.1cm, inner sep=0.5pt, text opacity=1, pattern=north east lines, distance=10pt, pattern color=red, fill opacity=0.3] (y_u_3) {\scriptsize $y_u^3$};
        \draw (0,1) node[draw,circle, minimum size=0.1cm, inner sep=0.5pt, text opacity=1, pattern=north east lines, distance=10pt, pattern color=green, fill opacity=0.3] (y_w_1) {\scriptsize $y_w^1$};
        \draw (9,0) node[draw,circle, minimum size=0.1cm, inner sep=0.5pt, text opacity=1, pattern=north east lines, distance=10pt, pattern color=blue, fill opacity=0.3] (y_v_1) {\scriptsize $y_v^1$};
        \draw (12,0) node[draw,circle, minimum size=0.1cm, inner sep=0.5pt, text opacity=1, pattern=north east lines, distance=10pt, pattern color=blue, fill opacity=0.3] (y_v_2) {\scriptsize $y_v^2$};
        \draw [->] (y_u_1) -- (y_u_3) node[midway,above left] {};
        \draw [->] (y_u_3) -- (y_u_2) node[midway,above left] {};
        \draw [->] (y_w_1) -- (y_u_3) node[midway,above left] {};
        \draw [->] (y_v_1) -- (y_v_2) node[midway,above left] {};
    \end{tikzpicture}
\end{figure}
\noindent Applying $E^T = \big\{ c_1 = \{ p_u^1 \},\ c_2 = \{p_v^2, p_u^2\},\ c_3 = \{ p_u^3 \},\ c_4 = \{ p_v^1 \},\ c_5 = \{ p_w^1 \} \big\}$ yields the posets:
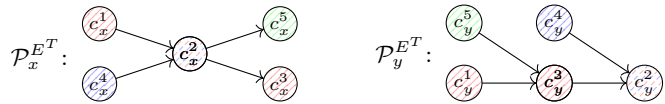
\begin{figure}[H]
    \centering
    \begin{tikzpicture}[scale=0.4, every node/.style={inner sep=0,outer sep=0}]
        \draw[fill=none] (-2, 0) node {\small $\mathcal{P}_x^{E^T}$:};
        \draw (0,1) node[draw,circle, minimum size=0.1cm, inner sep=0.5pt, text opacity=1, pattern=north east lines, distance=10pt, pattern color=red, fill opacity=0.3] (c_x_1) {\scriptsize $c_x^1$};
        \draw (3,0) node[draw,circle, minimum size=0.1cm, inner sep=0.5pt, text opacity=1] (c_x_2) {\scriptsize $c_x^2$};
        \begin{scope}[shift={(c_x_2.center)}]
            \clip (c_x_2.center) circle [radius=0.55cm];
            \draw[dash pattern= on 2pt off 4pt, blue, opacity=0.3] (-1.72cm, -1cm) -- (0.28cm, 1cm);
            \draw[dash pattern= on 2pt off 4pt, blue, opacity=0.3] (-1.47cm, -1cm) -- (0.53cm, 1cm);
            \draw[dash pattern= on 2pt off 4pt, blue, opacity=0.3] (-1.22cm, -1cm) -- (0.78cm, 1cm);
            \draw[dash pattern= on 2pt off 4pt, blue, opacity=0.3] (-0.97cm, -1cm) -- (1.03cm, 1cm);
            \draw[dash pattern= on 2pt off 4pt, blue, opacity=0.3] (-0.72cm, -1cm) -- (1.28cm, 1cm);
            \draw[dash pattern= on 2pt off 4pt, blue, opacity=0.3] (-0.47cm, -1cm) -- (1.53cm, 1cm);
            \draw[dash pattern= on 2pt off 4pt, blue, opacity=0.3] (-0.22cm, -1cm) -- (1.78cm, 1cm);
            \draw[dash pattern= on 2pt off 4pt, dash phase=3pt, red, opacity=0.3] (-1.72cm, -1cm) -- (0.28cm, 1cm);
            \draw[dash pattern= on 2pt off 4pt, dash phase=3pt, red, opacity=0.3] (-1.47cm, -1cm) -- (0.53cm, 1cm);
            \draw[dash pattern= on 2pt off 4pt, dash phase=3pt, red, opacity=0.3] (-1.22cm, -1cm) -- (0.78cm, 1cm);
            \draw[dash pattern= on 2pt off 4pt, dash phase=3pt, red, opacity=0.3] (-0.97cm, -1cm) -- (1.03cm, 1cm);
            \draw[dash pattern= on 2pt off 4pt, dash phase=3pt, red, opacity=0.3] (-0.72cm, -1cm) -- (1.28cm, 1cm);
            \draw[dash pattern= on 2pt off 4pt, dash phase=3pt, red, opacity=0.3] (-0.47cm, -1cm) -- (1.53cm, 1cm);
            \draw[dash pattern= on 2pt off 4pt, dash phase=3pt, red, opacity=0.3] (-0.22cm, -1cm) -- (1.78cm, 1cm);
        \end{scope}
        \draw (3,0) node[draw,circle, minimum size=0.1cm, inner sep=0.5pt, text opacity=1] (c_x_2_text) {\scriptsize $c_x^2$};
        \draw (6,-1) node[draw,circle, minimum size=0.1cm, inner sep=0.5pt, text opacity=1, pattern=north east lines, distance=10pt, pattern color=red, fill opacity=0.3] (c_x_3) {\scriptsize $c_x^3$};
        \draw (6,1) node[draw,circle, minimum size=0.1cm, inner sep=0.5pt, text opacity=1, pattern=north east lines, distance=10pt, pattern color=green, fill opacity=0.3] (c_x_5) {\scriptsize $c_x^5$};
        \draw (0,-1) node[draw,circle, minimum size=0.1cm, inner sep=0.5pt, text opacity=1, pattern=north east lines, distance=10pt, pattern color=blue, fill opacity=0.3] (c_x_4) {\scriptsize $c_x^4$};
        \draw [->] (c_x_1) -- (c_x_2) node[midway,above left] {};
        \draw [->] (c_x_2) -- (c_x_3) node[midway,above left] {};
        \draw [->] (c_x_2) -- (c_x_5) node[midway,above left] {};
        \draw [->] (c_x_4) -- (c_x_2) node[midway,above left] {};
    \end{tikzpicture}
    \hspace{0.8cm}
    \begin{tikzpicture}[scale=0.4, every node/.style={inner sep=0,outer sep=0}]
        \draw[fill=none] (-2, 0) node {\small $\mathcal{P}_y^{E^T}$:};
        \draw (0,-1) node[draw,circle, minimum size=0.1cm, inner sep=0.5pt, text opacity=1, pattern=north east lines, distance=10pt, pattern color=red, fill opacity=0.3] (c_y_1) {\scriptsize $c_y^1$};
        \draw (6,-1) node[draw,circle, minimum size=0.1cm, inner sep=0.5pt, text opacity=1] (c_y_2) {\scriptsize $c_y^2$};
        \begin{scope}[shift={(c_y_2.center)}]
            \clip (c_y_2.center) circle [radius=0.57cm];
            \draw[dash pattern= on 2pt off 4pt, dash phase=3pt, red, opacity=0.3] (-1.72cm, -1cm) -- (0.28cm, 1cm);
            \draw[dash pattern= on 2pt off 4pt, blue, opacity=0.3] (-1.72cm, -1cm) -- (0.28cm, 1cm);
            \draw[dash pattern= on 2pt off 4pt, blue, opacity=0.3] (-1.47cm, -1cm) -- (0.53cm, 1cm);
            \draw[dash pattern= on 2pt off 4pt, blue, opacity=0.3] (-1.22cm, -1cm) -- (0.78cm, 1cm);
            \draw[dash pattern= on 2pt off 4pt, blue, opacity=0.3] (-0.97cm, -1cm) -- (1.03cm, 1cm);
            \draw[dash pattern= on 2pt off 4pt, blue, opacity=0.3] (-0.72cm, -1cm) -- (1.28cm, 1cm);
            \draw[dash pattern= on 2pt off 4pt, blue, opacity=0.3] (-0.47cm, -1cm) -- (1.53cm, 1cm);
            \draw[dash pattern= on 2pt off 4pt, blue, opacity=0.3] (-0.22cm, -1cm) -- (1.78cm, 1cm);
            \draw[dash pattern= on 2pt off 4pt, dash phase=3pt, red, opacity=0.3] (-1.47cm, -1cm) -- (0.53cm, 1cm);
            \draw[dash pattern= on 2pt off 4pt, dash phase=3pt, red, opacity=0.3] (-1.22cm, -1cm) -- (0.78cm, 1cm);
            \draw[dash pattern= on 2pt off 4pt, dash phase=3pt, red, opacity=0.3] (-0.97cm, -1cm) -- (1.03cm, 1cm);
            \draw[dash pattern= on 2pt off 4pt, dash phase=3pt, red, opacity=0.3] (-0.72cm, -1cm) -- (1.28cm, 1cm);
            \draw[dash pattern= on 2pt off 4pt, dash phase=3pt, red, opacity=0.3] (-0.47cm, -1cm) -- (1.53cm, 1cm);
            \draw[dash pattern= on 2pt off 4pt, dash phase=3pt, red, opacity=0.3] (-0.22cm, -1cm) -- (1.78cm, 1cm);
        \end{scope}
        \draw (3,-1) node[draw,circle, minimum size=0.1cm, inner sep=0.5pt, text opacity=1] (c_y_2_text) {\scriptsize $c_y^2$};

        \draw (3,-1) node[draw,circle, minimum size=0.1cm, inner sep=0.5pt, text opacity=1, pattern=north east lines, distance=10pt, pattern color=red, fill opacity=0.3] (c_y_3) {\scriptsize $c_y^3$};
        \draw (0,1) node[draw,circle, minimum size=0.1cm, inner sep=0.5pt, text opacity=1, pattern=north east lines, distance=10pt, pattern color=green, fill opacity=0.3] (c_y_5) {\scriptsize $c_y^5$};
        \draw (3,1) node[draw,circle, minimum size=0.1cm, inner sep=0.5pt, text opacity=1, pattern=north east lines, distance=10pt, pattern color=blue, fill opacity=0.3] (c_y_4) {\scriptsize $c_y^4$};
        \draw [->] (c_y_1) -- (c_y_3) node[midway,above left] {};
        \draw [->] (c_y_3) -- (c_y_2) node[midway,above left] {};
        \draw [->] (c_y_4) -- (c_y_2) node[midway,above left] {};
        \draw [->] (c_y_5) -- (c_y_3) node[midway,above left] {};
    \end{tikzpicture}
\end{figure}
\end{exmp}

\paragraph{Pattern-Tree Occurrences.} As with corner-trees, we define \textit{pattern-tree occurrences}.

\begin{definition}
    An \textit{occurrence} $\varphi: p(T) \to p(\pi)$ of a pattern-tree $T$ in a permutation $\pi \in \Sn$ is a map whose \textit{image} $\varphi(p(T))$ \textit{conforms} to the constraints $\mathcal{C}(T)$.
\end{definition} 
An illustration of a pattern-tree occurrence is shown in \Cref{fig:pattern_tree}.
Note that, as with corner-trees, occurrence maps need not be injective.
We remark that some pattern-trees may have no occurrences, in any permutation $\pi \in \mathbb{S}_n$. 
For example, \raisebox{-6pt}{\begin{tikzpicture}[scale=0.2, every node/.style={inner sep=0,outer sep=0}]
	   \draw[fill=none] (-1.5, 0) node {\small $u$};
      \draw[fill=none] (15.5, 0) node {\small $v$};
      \draw (0,0) node[draw,circle, scale=0.3, minimum size=7mm, fill=black] (root) {};
	   \draw (14,0) node[draw,circle, scale=0.3, minimum size=7mm] (c) {};
	   \draw [thin, <-, scale=0.1] (c) -- (root) node[midway,below=0.5mm] {\small $p_u = p_v, x_u < x_v$};
\end{tikzpicture}} is infeasible.

\paragraph{Pattern-Tree Vectors.} As with corner-trees, one can associate a vector with every pattern-tree $T$, which is a formal integer linear combination of pattern-counts representing the number of occurrences of $T$ in any input permutation $\pi \in \Sn$. 

\begin{lemma}
    \label{lem:pattern_tree_vec}
    Let $T$ be a pattern-tree.
    The number of occurrences of $T$ in an input permutation $\pi \in \Sn$ is given by the following sum of pattern-counts, each of size at most $\Sigma(T)$:
    \[
        \#T(\pi) = \sum_{ E \supseteq E^T } \sum_{ \substack{ \sigma \in \mathcal{L}(\mathcal{P}_x^{E}) \\ \tau \in \mathcal{L}(\mathcal{P}_y^{E}) } } \pc{\left(\tau \sigma^{-1}\right)}{\pi}.
    \]
\end{lemma}
\begin{proof} 
Any occurrence $\varphi:p(T)\to p(\pi)$ assigns an $x$ coordinate in $[n]$, and $y$ coordinate $\pi(x)$, to each point $p \in p(T)$, in a way that agrees with both posets $\mathcal{P}_{x}^{T}$ and $\mathcal{P}_{y}^{T}$, and with the equivalence relation $E_T$ (i.e., equalities). The number of such assignments is the following:
\[
    \sum_{\substack{1\le x_{1} \le \cdots \le x_{\Sigma(T)}\le n \\ x_i = x_j\ \forall i,j:\ x_i \sim_{E_T} x_j }}\mathbbm{1}\left\{\text{\ensuremath{(x_{1},\ldots,x_{\Sigma(T)})} satisfies \ensuremath{\mathcal{P}_{x}^{T}}}\right\}\cdot\mathbbm{1}\left\{\text{\ensuremath{(\pi(x_{1}),\ldots,\pi(x_{\Sigma(T)}))} satisfies \ensuremath{\mathcal{P}_{y}^{T}}}\right\},
\]
where we say that $(x_{1},\ldots,x_{\Sigma(T)})$ satisfies $\mathcal{P}_{x}^{T}$,
if whenever we replace the $x$ coordinate of the $i$-th point in $p(T)$
(according to some arbitrary \textit{fixed} order on $p(T)$) with $x_i$,
all the inequalities defined by $\mathcal{P}_{x}^{T}$ hold true. Likewise the $y$ coordinates.

Any valid choice of $x_{1},\ldots,x_{\Sigma(T)}$
defines an equivalence relation $E\supseteq E^{T}$, determined by which coordinates are equal.
Let $a_1 < \cdots < a_k$ be the distinct $x$ coordinates among $x_1, \ldots, x_{\Sigma(T)}$,
where $k \eqdef |E| \le \Sigma(T)$ ($|E|$ is the number of equivalence classes in $E$).
Let $\sigma\in\mathbb{S}_{k}$ be the permutation where the $i$-th equivalence class
(assuming some arbitrary fixed order) is assigned coordinate $a_{\sigma(i)}$.
That is, all points of $p(T)$ in the $i$-th equivalence class, are mapped to the permutation
point whose $x$ coordinate is $a_{\sigma(i)}$.
Under this notation, 
\[
    (x_{1},\ldots,x_{\Sigma(T)}) \text{ satisfies }\mathcal{P}_{x}^{T} \iff (a_{\sigma(1)},\ldots,a_{\sigma(k)}) \text{ satisfies }\mathcal{P}_{x}^{E},
\]
and similarly for $(\pi(a_{\sigma(1)}),\ldots,\pi(a_{\sigma(k)}))$ and $\mathcal{P}_{y}^{E}$. 

\ \\ \noindent By rearranging the previous sum, we obtain: 
\begin{align*}
    \#T(\pi) &= \sum_{E\supseteq E^{T}}\sum_{\sigma\in\mathcal{L}(\mathcal{P}_{x}^{E})}\sum_{1 \le a_1 < \dots < a_{k} \le n}\mathbbm{1} \left\{ (\pi(a_{\sigma(1)}),\ldots,\pi(a_{\sigma(k)})) \text{ satisfies } \mathcal{P}_{y}^{E} \right\} \\
    &= \sum_{E\supseteq E^{T}}\sum_{\substack{\sigma\in\mathcal{L}(\mathcal{P}_{x}^{E}) \\ \tau  \in \mathcal{L}(\mathcal{P}_{y}^{E})}}\sum_{1 \le a_1 < \dots < a_{k} \le n}\mathbbm{1} \left\{ \pi[a_{\sigma(1)},\ldots,a_{\sigma(k)}] \cong \tau \right\} \\
    &= \sum_{E\supseteq E^{T}}\sum_{\substack{\sigma\in\mathcal{L}(\mathcal{P}_{x}^{E}) \\ \tau  \in \mathcal{L}(\mathcal{P}_{y}^{E})}}\sum_{1 \le a_1 < \dots < a_{k} \le n}\mathbbm{1} \left\{ \pi[a_{1},\ldots,a_{k}] \cong \tau \cdot \sigma^{-1} \right\}
\end{align*}
and the latter sum simply counts the occurrences of the pattern $\tau \sigma^{-1}$ in $\pi$, as required.
\end{proof}

\paragraph{Evaluating a Pattern-Tree.} It remains to construct an evaluation algorithm for the vector of a pattern-tree.
To present our algorithm, we require some notation.
\begin{enumerate}
    \item \underline{Points}: To every set of points $S \eqdef \{ s_1, \dots, s_r \} \subseteq p(\pi)$,
    where $\pi \in \Sn$ is a permutation and $(s_1)_x < \dots < (s_r)_x$,
    we associate a $2r$-dimensional point,
    \[
        p(S) \eqdef \big( (s_1)_x, \dots, (s_r)_x, (s_1)_y, \dots, (s_r)_y \big) \in [n]^{2r}
    \]

    \item \underline{Rectangles}: To every combination of an edge $(u \to v) \in E(T)$ in a pattern-tree $T$,
where $v$ and $u$ are of sizes $d$ and $r$ respectively,
and set of points $S \eqdef \{ s_1, \dots, s_r \} \subseteq p(\pi)$, 
we associate a $2d$-dimensional rectangle,
\[
    \mathcal{R}_{uv}^S \eqdef \mathcal{R}_{uv}^{S,x} \times \mathcal{R}_{uv}^{S,y} \subseteq [n]^{2d},
    \text{ where } \mathcal{R}_{uv}^{S,x}, \mathcal{R}_{uv}^{S,y} \subseteq [n]^d.
\]
The $i$-th segment of $\mathcal{R}_{uv}^{S,x}$ contains the $x$ coordinates that $x_v^i$ can take under the constraints of $u\to v$, when $x_u^j$ is assigned $(s_j)_x$. Namely, the intersection of the following segments:
\[
    \underbrace{\bigcap_{j : (p_u^j, p_v^i) \in E_{uv}} \left\{ (s_j)_x \right\}}_{\text{equals}},\; \underbrace{\bigcap_{j: (x_v^i < x_u^j) \in \mathcal{P}_{uv}^x } \left\{1, \dots, (s_j)_x - 1\right\}}_{\text{less-than}},\; \underbrace{\bigcap_{j: (x_v^i > x_u^j) \in \mathcal{P}_{uv}^x } \left\{ (s_j)_x+1, \dots, n \right\}}_{\text{greater-than}}
\]
The $y$-segments are similarly defined.    
\end{enumerate}

Observe that the rectangle $\mathcal{R}_{uv}^S$ is the set of permissible
locations for the points $p_v$, subject to the edge-constraints on the edge $u \to v$,
when the points $p_u$ are mapped to $p(S)$. That is, it enforces both the equalities (left)
and inequalities (centre, right) written on the edge $u \to v$.\ \\

The evaluation algorithm now follows.

\begin{algorithm}[H]
\caption{Bottom-Up Evaluation of Pattern-Tree Vector}
\label{alg:bottom_up_pattern_tree}
\vspace{0.2cm}\hspace*{\algorithmicindent} \textbf{Input:} A pattern-tree $T$, and a permutation $\pi \in \mathbb{S}_n$. 
\begin{enumerate}
    \item Traverse the vertices of $T$ in post-order. For every vertex $u$ labeled by $\tau_u \in \mathbb{S}_r$:
    \begin{enumerate}
        \item Construct a new (empty) rectangle-tree $\mathcal{T}_u$ of dimension $2r$.
        \item Iterate over all sets $S \eqdef \{ s_1, \dots, s_r\} \subseteq p(\pi)$. If $S \cong \tau_u$, then:
        \begin{enumerate}
            \item For every child $v$ of $u$, issue the query $\mathcal{T}_v(\mathcal{R}_{uv}^S)$.
            \item Add the weight $\prod_{u \to v} \mathcal{T}_v(\mathcal{R}_{uv}^S)$ (or $1$, if $u$ is a leaf) to point $p(S)$ in $\mathcal{T}_u$.
        \end{enumerate}
    \end{enumerate}
    \item Return the answer to the query $\mathcal{T}_{z}(\mathcal{R})$, where $z$ is the root of $T$, and $\mathcal{R} = [n]^{2 |\tau_z|}$.
\end{enumerate}
\end{algorithm}
\begin{theorem}
    \label{thm:computing_pattern_trees}
    Let $T$ be a pattern-tree of constant total size, and let $\pi \in \mathbb{S}_n$ be a permutation.
    The vector of $T$ can be evaluated over $\pi$ in $\Otilde{n^{s(T)}}$ time, where $s(T)$ is the \textit{maximum size}.\footnote{
        The space-complexity is also $\Otilde{n^{s(T)}}$,
        since at every vertex of size $r$, we insert $\le \binom{n}{r}$ points to a rectangle-tree.
    }
\end{theorem}
\begin{proof}
    The running time of \Cref{alg:bottom_up_pattern_tree} is $\Otilde{n^{s(T)}}$, since every operation takes $\Otilde{1}$ time (recall that $\Sigma(T)=\mathcal{O}(1)$),
    except step (1b.), which we perform in time $\mathcal{O}(n^r)$, by trivial enumeration.
    It remains to prove its correctness. We do so, by induction on the height of the tree.
    
    Let $u \in V(T)$ be a vertex, let $\tau_u \in \mathbb{S}_r$ be its permutation label, and let $T_{\le u}$ be the sub-tree rooted at $u$.
    We claim that for every $S \eqdef \{ s_1, \dots, s_r \} \subseteq p(\pi)$,
    the weight of $p(S)$ in $\mathcal{T}_u$ is the number of occurrences $\varphi: p(T_{\le u}) \to p(\pi)$ in which
    the points $p(u)$ are mapped to $S$.
    That is, for every $1 \le i \le r$, it holds that $\varphi(p_u^i) = s_i$. 

    In the base-case, $u$ is a leaf, and \Cref{alg:bottom_up_pattern_tree} simply enumerates over all sets $S$ of cardinality $r$,
    adding weight $1$ whenever $S \cong \tau_u$. So the claim holds.
    For the inductive step, let $u$ be an internal vertex.
    For every child $v$ of $u$, by the induction hypothesis, the query $\mathcal{T}_v(\mathcal{R}_{uv}^S)$
    counts the number of occurrences $\varphi_v: p(T_{\le v}) \to p(\pi)$ in which there exists a point $A \in \mathcal{R}_{uv}^S$
    such that $\varphi_v(p_u^i) = a_i$, for every $i$. That is, the number of occurrences of the tree
    in which we add the $u$ as the root to the tree $T_{\le v}$,
    where the occurrence maps $p_u^i$ to $s_i$ for every $i \in [r]$.
    These occurrences are \textit{independent} for every child $v$ of $u$,
    therefore picking any combination of them yields a new occurrence of $T_{\le u}$ in $\pi$, 
    the total number of which is indeed the product $\prod_{u \to v} \mathcal{T}_v(\mathcal{R}_{uv}^S)$.

    The proof now follows, as in the rectangle-tree $\mathcal{T}_z$ corresponding to the root $z$ of $T$,
    every point $S$ has weight which is the number of occurrences of $T$ in $\pi$ in which $p_z$ is mapped to $S$.
    Therefore, the sum of all points in $\mathcal{T}_z$ yields the total number of occurrences.
\end{proof}

\begin{rem}
    The algorithm presented in \Cref{thm:computing_pattern_trees} is not necessarily the most efficient way to compute the vector of a pattern-tree, for several reasons.
    Firstly, many trees may correspond to the same vector, and these trees need not have the same maximum size. For example, both \\
    \vspace{-0.3cm}
    \begin{center}
    \begin{tikzpicture}[scale=0.2, every node/.style={inner sep=0,outer sep=0}]
    	   \draw[fill=none] (-1.3, 0) node {\small $u$};
          \draw[fill=none] (16.3, 0) node {\small $v$};
          \draw (0,0) node[draw,circle, scale=0.3, minimum size=7mm, fill=black] (root) {};
    	   \draw (15,0) node[draw,circle, scale=0.3, minimum size=7mm] (c) {};
    	   \draw [thin, <-, scale=0.1] (c) -- (root) node[midway,above=0.3mm] {\small $x_u < x_v, y_u < y_v$};
    \end{tikzpicture} \ \ and\ \  \circled{$\mathtt{12}$}
    \end{center}
    correspond to the vector $\pce{12}$.
    Secondly, as we will see in \Cref{sect:5_prof_alg}, there exist vectors for which bespoke \textit{efficient} algorithms can be constructed,
    whose running time is strictly smaller than the maximum size of \textit{any} pattern-tree with the same vector. 
\end{rem}

\subsection{\texorpdfstring{$\Otilde{n^2}$}{O(n\^2)} Algorithm for the \texorpdfstring{$k$}{k}-Profile, for \texorpdfstring{$1 \le k \le 7$}{1 <= k <= 7}}
\label{subsect:5_to_7_prof}

The corner-trees of \cite{even2021counting} are very efficiently computable.
However, asymptotically, there are quite few of them:
the number of rooted unlabeled trees over $k$ vertices is only exponential in $k$
(see, e.g., \cite{knuth1997art} for a more accurate estimate),
and clearly so is the number of corner-tree edge labels.
Therefore, as $k \to \infty$,
even if asymptotically almost all corner-trees vectors were linearly independent over $\mathbb{S}_{\le k}$, they would nevertheless contribute only a \textit{negligible} proportion with respect to the full dimension, $|\mathbb{S}_{\le k}| = \sum_{r=1}^k r!$.

In contrast, it is not hard to see that pattern-trees are \textit{fully expressive}:
for every pattern $\tau \in \mathbb{S}_k$,
there exists a pattern-tree $T$ with $s(T)=k$, whose vector is precisely that pattern
(in fact, $s(T)=\lceil k/2 \rceil$ suffices, see \Cref{subsect:k_over_2_alg}).
To design \textit{efficient} algorithms for the $k$-profile,
we are interested in finding families of pattern-trees of \textit{least maximum size},
whose corresponding vectors are linearly independent.

In \cite{even2021counting}, corner-trees (i.e., pattern-tree of maximum size $1$) over $k$ vertices
were shown to have full rank over $\QQ^{\mathbb{S}_{\le k}}$ for $k=3$, and
in the cases $k=4$ and $k=5$, the subspaces spanned by them,
restricted to $\mathbb{S}_4$ and $\mathbb{S}_5$,
were found to be of dimensions only $23$ and $100$, respectively.
Here, we show that for $k \le 7$, pattern-trees of maximum size $\le 2$ suffice. 

\begin{proof}
[Proof of \Cref{thm:fast_le_7_profile}.] Let $\mathbb{S} \eqdef \bigsqcup_{k=1}^7 \mathbb{S}_k$.
By enumeration (see \Cref{sect:enumeration_pattern_trees}),
there exists a family of $\sum_{k=1}^7 k! = 5913$ pattern-trees of maximum size at most $2$ and total size at most $7$,
whose vectors are linearly independent over $\QQ^{\mathbb{S}}$.
Let $A \in \QQ^{\mathbb{S} \times \mathbb{S}}$ be the matrix whose rows are these vectors,
and let $A^{-1} \in \QQ^{\mathbb{S} \times \mathbb{S}}$ be its inverse.
$A$ may be computed ahead of time, as can its inverse, for example using Bareiss' algorithm \cite{bareiss1968sylvester}.
Using \Cref{thm:computing_pattern_trees}, evaluate every row of $A$ over $\pi$ in time $\Otilde{n^2}$.
This yields a vector $v \in \ZZ^{\mathbb{S}}$,
and the $k$-profiles of $\pi$, for $k \le 7$, are obtained by computing $A^{-1} v$.    
\end{proof}

\subsection{The case \texorpdfstring{$k=8$}{k=8}}
\label{subsect:s_8_not_spanned}
Do pattern-trees over at most $8$ points, and with $s(T) \le 2$, have full dimension for $\mathbb{S}_{\le 8}$?
Using a computer program, we exhaustively enumerate all pattern-trees with the following properties,\footnote{
See \Cref{sect:enumeration_pattern_trees} for a description of the enumeration process.
For $k=8$, this yields a matrix with $|\mathbb{S}_8|=8!$ columns, and $|\mathbb{S}_8 \times \mathbb{S}_8 \times \{ T_\lambda \}| \approx 2^{37}$ rows.
We remark that we explicitly do not consider pattern-trees over \textit{more} than $8$ points, and trees whose edges are labeled by equalities.
Whether this is without loss of generality, i.e., could their inclusion increase the rank, is unknown to us.
}
\begin{enumerate}
    \item Every tree has $|p(T)| = 8$ points, and maximum size $s(T) \le 2$. 
    \item No edge is labeled with an equality.
\end{enumerate}

In \cite{even2021counting} it was shown that pattern-trees with $4$ vertices and maximum size $1$ (corner-trees)
span a subspace of dimension only $|\mathbb{S}_4|-1 = 23$, when restricted to $\mathbb{S}_4$.
Our pattern-trees extend this result: 
the subspace spanned by the above family of pattern-trees,
with $8$ points and maximum size $\le 2$,
is of dimension exactly $|\mathbb{S}_8|-1 = 40319$, when restricted to $\mathbb{S}_8$.
The two vectors
spanning the orthogonal complements of the subspaces for $\mathbb{S}_4$ and $\mathbb{S}_8$,
$v_4 \in \QQ^{\mathbb{S}_4}$ and $v_8 \in \QQ^{\mathbb{S}_8}$ respectively,
bear striking resemblance, as we detail below.

One of the central components in the $4$-profile algorithm of \cite{dudek2020counting} is the 
classification of patterns in $\mathbb{S}_4$ into two sets: \textit{trivial} and \textit{non-trivial}.
A pattern $\tau \in \mathbb{S}_4$ is called non-trivial
if its four points appear each in a different quadrant of the square $[4]^2$.
A pattern is called trivial otherwise.
An occurrence of a non-trivial permutation $\tau$ in a permutation $\pi \in \mathbb{S}_n$, in which
the points of $\tau$ appear in the four quadrants of $[n]^2$, is called $4$-partite.

There are $16$ non-trivial patterns in $\mathbb{S}_4$, and they exactly form the support of the vector $v_4$.
Half appear with magnitude $1$, and half with magnitude $-1$. Clearly this implies that all trivial patterns
can be counted in quasi-linear time (see \Cref{prop:compute_corner_tree}).
In fact, Dudek and Gawrychowski \cite{dudek2020counting} observe that the only ``hard case'' in computing the $4$-profile
is counting the $4$-partite occurrences of the non-trivial patterns, 
and prove a  bidirectional reduction between enumerating such occurrences, and counting $4$-cycles in sparse graphs.

At the heart of their algorithm for $4$-partite occurrences
lies an observation regarding the symmetries of the non-trivial permutations:
they are closed both under the action of $D_4 \acts \mathbb{S}_4$,
\textit{and} the action of swapping the first two points
(i.e., reflecting the \textit{left half} of the square horizontally). \ \\

We extend all of the above characterisations to $\mathbb{S}_8$, as follows.
Say that a pattern $\tau \in \mathbb{S}_8$ is \textit{non-trivial} if it
satisfies the following:

\begin{enumerate}
    \item Each quadrant contains exactly two points.
    \item The number of ascending (resp. descending) pairs in the four quadrants is odd.
    \item Every \textit{half} (top, bottom, left and right) of $\tau$ is a non-trivial permutation in $\mathbb{S}_4$.
\end{enumerate}
We call a pattern trivial otherwise.

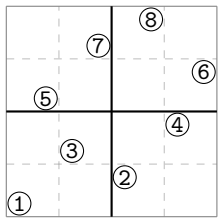
\begin{figure}[H]
    \centering
    \begin{tikzpicture}[scale=0.35]

        \foreach \x in {0, 4, 8} {
            \draw[color=gray] (\x, 0) -- (\x, 8);
        }
        \foreach \x in {2,6} {
            \draw[color=gray, opacity=0.5, dashed] (\x, 0) -- (\x, 8);
        }

        \foreach \y in {0, 4, 8} {
            \draw[color=gray] (0, \y) -- (8, \y);
        }
        \foreach \y in {2,6} {
            \draw[color=gray, opacity=0.5, dashed] (0, \y) -- (8, \y);
        }

        \draw[color=black, thick] (0, 4) -- (8, 4);
        \draw[color=black, thick] (4, 0) -- (4, 8);
        
        \draw (0.5,0.5) node[draw,circle, minimum size=0.1cm, inner sep=0.5pt, text opacity=1] (1) {\small $\mathtt{1}$};
        \draw (1.5,4.5) node[draw,circle, minimum size=0.1cm, inner sep=0.5pt, text opacity=1] (5) {\small $\mathtt{5}$};
        \draw (2.5,2.5) node[draw,circle, minimum size=0.1cm, inner sep=0.5pt, text opacity=1] (3) {\small $\mathtt{3}$};
        \draw (3.5,6.5) node[draw,circle, minimum size=0.1cm, inner sep=0.5pt, text opacity=1] (7) {\small $\mathtt{7}$};
        \draw (4.5,1.5) node[draw,circle, minimum size=0.1cm, inner sep=0.5pt, text opacity=1] (2) {\small $\mathtt{2}$};
        \draw (5.5,7.5) node[draw,circle, minimum size=0.1cm, inner sep=0.5pt, text opacity=1] (8) {\small $\mathtt{8}$};
        \draw (6.5,3.5) node[draw,circle, minimum size=0.1cm, inner sep=0.5pt, text opacity=1] (4) {\small $\mathtt{4}$};
        \draw (7.5,5.5) node[draw,circle, minimum size=0.1cm, inner sep=0.5pt, text opacity=1] (6) {\small $\mathtt{6}$};
    \end{tikzpicture}
    \caption{A non-trivial pattern $\tau = \mathtt{15372846} \in \mathbb{S}_8$. 
    There are three ascending pairs in its quadrants, and one descending pair.
    Its halves are order-isomorphic to the non-trivial permutations, $\mathtt{1342}$ (top), $\mathtt{1324}$ (bottom, left) and $\mathtt{1423}$ (right).}
    \label{fig:nontrivial-8}
\end{figure}

There are $2048$ non-trivial permutations in $\mathbb{S}_8$.
The support of the vector $v_8$ consists exactly of the non-trivial patterns of $\mathbb{S}_8$.
Again, half appear with magnitude $1$, and the other half (which are the vertical or
horizontal reflections of the first set) appear with magnitude $-1$.
This of course implies that all trivial patterns can be counted in $\Otilde{n^2}$-time. 
One can further extend the analogy to \cite{dudek2020counting} by noting that all
non-trivial patterns are closed under the action of $D_4 \acts \mathbb{S}_8$, \textit{and} the actions
of swapping the first two elements, or the first two pairs
(i.e., reflecting the left quarter-strip, or the left half of the square horizontally).

We find the emergence of this ``pattern'' of non-trivial permutations and their relation to pattern-trees
to be highly interesting. In fact, in direct analogy to \cite{dudek2020counting},
we conjecture that, as with $\mathbb{S}_4$, 
the occurrences of non-trivial patterns $\tau \in \mathbb{S}_8$ in a permutation $\pi \in \mathbb{S}_n$,
in which the points of $\tau$
appear in the above configuration within the square $[n]^2$,
constitute the ``hard case'' for computing the $8$-profile.
Settling this question, as well as understanding the (possibly algebraic) relation between pattern-trees
of maximum size $\le s$, and non-trivial permutations, are left as open questions.\footnote{
Another possible extension of the analogy with regards to \cite{dudek2020counting} is the following: 
it is known that for $3 \le k \le 7$, the number of length-$k$ cycles in an $n$-vertex graph can be counted
in time $\Otilde{n^\omega}$ \cite{alon1997finding}, where $\omega$ is the exponent of matrix multiplication.
Whether this cutoff at $k=8$ relates to the $8$-profile problem is unknown to us.
}

\subsection{\texorpdfstring{$\Otilde{n^{\lceil k / 2 \rceil}}$}{O(nk/2)} Algorithm for the \texorpdfstring{$k$}{k}-Profile}
\label{subsect:k_over_2_alg}

We end this section by considering the problem of computing the $k$-profile via pattern-trees, for arbitrary (fixed) $k$.
In the following proposition, we show that families of pattern-trees of maximal size $s(T)=\lceil k / 2 \rceil$
suffice for computing the $k$-profile, through \Cref{alg:bottom_up_pattern_tree}.
See \Cref{sect:discussion} for a discussion on the relationship between $s(T)$ and $k$.\

\begin{proposition}
    \label{prop:k_over_2_family}
    Let $\pi \in \Sn$ be an input permutation, and let $k\ge 2$ be a fixed integer. The $k$-profile of $\pi$ can be computed in $\Otilde{n^{\lceil k / 2 \rceil}}$ time.
\end{proposition}
\begin{proof}
    As $k$ is fixed, it suffices to compute $\pc{\tau}{\pi}$ in $\Otilde{n^{\lceil k/2 \rceil}}$ time,
    for every pattern $\tau \in \mathbb{S}_k$.
    Let $\tau\in\mathbb{S}_k$ be a pattern, and let $S_1\sqcup S_2=p(\tau)$ be a partition of the points of $\tau$.
    Let $\sigma_1$ and $\sigma_2$ be the patterns for which $S_1\cong \sigma_1$ and $S_2\cong \sigma_2$.
    Consider a pattern-tree $T$ with two vertices labeled $\sigma_1$ and $\sigma_2$,
    and the edge between them constraining every pair of points according to $p(\tau)$.
    Notice that any constraint can be fixed by either a vertex or an edge.
    Therefore, there is a one-to-one correspondence between occurrences of $\tau$ and of $T$,
    so $\pcn{T}{\pi}=\pcn{\tau}{\pi}$.
    We can take $S_1,S_2$ such that the cardinality of no part exceeds $\lceil k/2 \rceil$, and the claim now follows from \Cref{thm:computing_pattern_trees}.
\end{proof}

\section{\texorpdfstring{$\Otilde{n^{7/4}}$}{O(n 7/4)} Algorithm for the \texorpdfstring{$5$}{5}-Profile}
\label{sect:5_prof_alg}

In \Cref{sect:pattern_trees}, we recalled that pattern-trees of maximum size $1$ (i.e., corner-trees)
have full rational rank for $\mathbb{S}_{\le 3}$ \cite{even2021counting},
and proved that trees of maximal size at most $2$ have full rank for $\mathbb{S}_{\le 7}$ (see \Cref{thm:fast_le_7_profile}). 
Therefore, up to $k=3$, the $k$-profile of an $n$-element permutation can be computed in $\Otilde{n}$ time,
and up to $k=7$, it is computable in $\Otilde{n^2}$ time.
This naturally raises the question: is there a sub-quadratic time algorithm for these cases, where $k \ge 4$?
We prove the following.

\begin{thm}
    \label{thm:fast-5-prof}
    The $5$-profile of any $n$-element permutation can be computed in time $\Otilde{n^{7/4}}$.
\end{thm}

We remark that the case $k=4$ has been extensively studied in \cite{dudek2020counting} and \cite{even2021counting}.
There, they construct sub-quadratic algorithms of complexities $\Oh{n^{1.478}}$ and $\Otilde{n^{3/2}}$, respectively.

\subsection{Marked and Weighted Patterns}

For the proof of \Cref{thm:fast-5-prof}, we introduce the following notation.

\paragraph{Marked Patterns.}
A \emph{marked pattern} is a pattern $\tau\in\mathbb{S}_k$ associated with an index $1\le j \le k$.
We say that a marked pattern $\tau$ occurs at index $1\le i\le n$ in $\pi\in\Sn$,
if there exists an occurrence of $\tau$ in $\pi$, in which the $j$-th $x$-coordinate is $i$.
When the marked pattern $\tau$ is short, we underline the $j$-th index to indicate that 
marked index.
For instance, $\mathtt{\underline{2}1}$ occurs in $\mathtt{132}$ at index $2$.

The \emph{marked pattern count} is a $2$-dimensional rectangle-tree containing the points $p(\pi)$,
in which the weight of every point $(i,\pi(i))$ is the number of marked pattern occurrences at position $i$.
For example, the tree $\mathcal{T}_2$ appearing in \Cref{prop:count-123k} is precisely the marked pattern count
$\pc{1\underline{2}}{\pi}$.

\paragraph{Weighted Pattern Counts.}
Let $\pi\in\Sn$ and let $w_1,\ldots,w_k:[n]\to\ZZ$ be \emph{weight functions}, where $k\ge 1$ is a fixed integer. The \emph{weighted pattern count} of $\tau\in\mathbb{S}_k$ in $\pi$, denoted $\pcw{\tau}{\pi}$, is the sum of $\prod_{j=1}^k w_j(i_j)$ over all occurrences $1\le i_1 < \cdots < i_k \le n$ of $\tau$ in $\pi$.
In other words, we count occurrences where every point has weight depending on its position, rather than $1$ as usual.
\ \\

The two concepts of marked patterns and weighted patterns can be combined in a straightforward way:
the \emph{weighted marked pattern count} is once again defined as a $2$-dimensional rectangle-tree,
as with marked pattern counts, but where 
now the number of occurrences for each point $(i,\pi(i))$ is appropriately weighted.

\subsection{An Improvement to the Bottom-Up Algorithm}
\label{subsect:bottom-up-improve}

Recall that \Cref{alg:bottom_up_pattern_tree} has time complexity $\Otilde{n^{s(T)}}$, where $s(T)$ is an integer.
As we seek sub-quadratic algorithms,
and since trees of $s(T)=1$ (i.e., corner-trees) do not have full rank
for $\mathbb{S}_{\le 5}$, we take an alternative approach. \ \\

\noindent Let $u$ be vertex of a pattern-tree $T$, labeled by some permutation $\tau_u \in \mathbb{S}_r$, such that:
\begin{enumerate}
    \item The incoming edge to $u$ (if any) conditions on a single point of $u$, say $p_u^l$.
    \item Each outgoing edge of $u$ (if any) is labeled by a single equality to a point of $u$.
\end{enumerate}
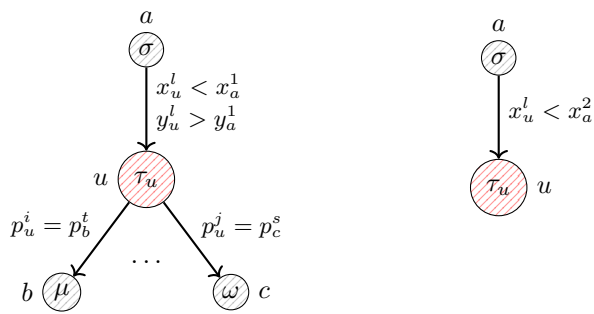
\begin{figure}[H]
    \centering
    \begin{tikzpicture}[scale=0.75]
      \draw (0,2.3) node[draw,circle, minimum size=0.3cm, inner sep=2pt,pattern=north east lines, distance=10pt, pattern color=gray, fill opacity=0.4, text opacity=1, label=above:{$a$}] (a) {$\sigma$};
      
      \draw (0,0) node[draw,circle, minimum size=0.75cm, inner sep=2pt,pattern=north east lines, distance=10pt, pattern color=red, fill opacity=0.4, text opacity=1, label=left:{$u$}] (u) {$\tau_u$};
      \draw (-1.5,-2) node[draw,circle, minimum size=0.3cm, inner sep=2pt,pattern=north east lines, distance=10pt, pattern color=gray, fill opacity=0.4, text opacity=1, label=left:{$b$}] (b) {$\mu$};
      \draw (1.5,-2) node[draw,circle, minimum size=0.3cm, inner sep=2pt,pattern=north east lines, distance=10pt, pattern color=gray, fill opacity=0.4, text opacity=1, label=right:{$c$}] (c) {$\omega$};
      
      \draw [thick, ->] (u) -- (b) node[midway,above left, yshift=-0.1cm] {\small $p_u^i=p_b^t$};
      \draw [thick, ->] (u) -- (c) node[midway,above right, yshift=-0.1cm] {\small $p_u^j=p_c^s$};
      \draw [thick, <-] (u) -- (a) node[align=left,midway,right, yshift=0.1cm] {\small $x_u^l<x_a^1$ \\ \small $y_u^l > y_a^1$};
      \node [yshift=0.4cm] (centre_dots) at ($(b)!0.5!(c)$) {$\ldots$};
    \end{tikzpicture}
    \hspace{1cm}
    \begin{tikzpicture}[scale=0.75]
      \draw (0,2.3) node[draw,circle, minimum size=0.3cm, inner sep=2pt,pattern=north east lines, distance=10pt, pattern color=gray, fill opacity=0.4, text opacity=1, label=above:{$a$}] (a) {$\sigma$};
      
      \draw (0,0) node[draw,circle, minimum size=0.75cm, inner sep=2pt,pattern=north east lines, distance=10pt, pattern color=red, fill opacity=0.4, text opacity=1, label=right:{$u$}] (u) {$\tau_u$};
      \draw (-1.5,-2) node[circle, minimum size=0.3cm, inner sep=2pt,pattern=north east lines, distance=10pt, pattern color=gray, fill opacity=0, text opacity=1, label=left:{}] (b) {};
      \draw [thick, <-] (u) -- (a) node[midway,right, yshift=0.1cm] {\small $x_u^l<x_a^2$};
    \end{tikzpicture}
    \caption{Two ``gadgets'' in a pattern-tree. The left corresponds to a weighted marked pattern count of $\tau_u$, marked at $l$. The right corresponds to a (unweighted) marked pattern count of $\tau_u$.}
    \label{fig:5-profile-special}
\end{figure}

Suppose that, for the permutation $\tau_u \in \mathbb{S}_r$ and index $l \in [r]$,
and given a set of weight functions $\{w_j\}_j$,\footnote{
We assume that for every weight function $w_j: [n] \to \ZZ$, the value $w_j(a)$ can be computed in $\Otilde{1}$-time.}
we are able to construct a $2$-dimensional rectangle-tree
representing the \textit{weighted marked pattern-count}, $\pcnw{\tau_u}{\pi}$ marked at $l$.
Then, we claim that one can \textit{modify} \Cref{alg:bottom_up_pattern_tree} by
replacing the rectangle-tree $\mathcal{T}_u$ associated with $u$, with the weighted marked pattern count
of $\tau$, marked at $l$, for a particular choice of weight functions.
Concretely, we make the following modifications in \Cref{alg:bottom_up_pattern_tree}: 
 
\paragraph{Traversing $u$.}
Instead of the routine operation of \Cref{alg:bottom_up_pattern_tree}, when $u$ is visited 
we compute a weighted $l$-marked pattern count $\pcnw{\tau_u}{\pi}$, abbreviated as $\mathcal{T}'_u$, with the following weights:
for every point $p_u^j$, define a weight function $w_j: [n] \to \ZZ$ by
\[
w_j(a) \eqdef \prod_{\substack{u\to v \\ \text{$p_u^j$ constrained}}} \mathcal{T}_v({\mathcal{R}}_{v}^{i})
\]
where for an edge $u\to v$ labeled $p_u^j=p_v^i$, we define ${\mathcal{R}}_{v}^{i}$
as the rectangle in which the $i$-th $x$-segment is $\{a\}$ and all other segments are unconstrained
(if $p_u^j$ is not constrained by any outgoing edges, set its weight function to $1$).
By the invariant of \Cref{alg:bottom_up_pattern_tree},
the query $\mathcal{T}_v(\mathcal{R}_v^i)$ counts the number of occurrences of $T_{\le v}$ in $\pi$ such that $x_v^i=a$.
Therefore, the resulting tree $\mathcal{T}'_u$ contains, at every point $(i,\pi(i))$,
the number of occurrences of $T_{\le u}$ in $\pi$ such that $x_u^l=i$. 

\paragraph{Querying $u$.} In \Cref{alg:bottom_up_pattern_tree} we query the rectangle-trees of vertices 
in two scenarios:
\begin{enumerate}
    \item \underline{If $u$ is an internal vertex}: In the original formulation of \Cref{alg:bottom_up_pattern_tree},
    when the parent $z$ of $u$ is visited, we issue queries of the form $\mathcal{T}_u(\mathcal{R}^S_{zu})$,
    for pointsets $S \subseteq p(\pi)$. As the edge $z \to u$ only constrains $p_u^l$, the
    rectangles $\mathcal{R}^S_{zu}$ are \textit{degenerate}, i.e., all of their segments are complete, except
    the two segments corresponding to $p_u^l$.
    These queries can be answered by $\mathcal{T}'_u(\mathcal{R}_u^l)$, where $\mathcal{R}_u^l \subseteq [n]^2$ is the
    $2$-dimensional projection of $\mathcal{R}_{zu}^S$ onto those two segments.
    \item \underline{If $u$ is the root}: The final step of the algorithm performs the full rectangle query $\mathcal{T}_u ([n]^{2|\tau_u|})$, which counts the occurrences of $T_{\le u}=T$ in all of $\pi$. This can be answered by the full rectangle query $\mathcal{T}'_u ([n]^2)$. 
\end{enumerate}

As for the correctness of this modification to \Cref{alg:bottom_up_pattern_tree}, it remains to show that the new queries return the same values as the original ones. Let $\mathcal{R}$ be some rectangle query to $\mathcal{T}_u$. The value of $\mathcal{T}_u(\mathcal{R})$ is the number of occurrences of $T_{\le u}$ in $\pi$, constrained to the coordinates allowed by $\mathcal{R}$. Since in both cases, all segments in $\mathcal{R}$ are complete except possibly those corresponding to $p_u^l$, this counts the occurrences of $T_{\le u}$ in $\pi$ constrained only to $p_u^l\in\mathcal{R}'$, for a $2$-dimensional projection $\mathcal{R}'$ of $\mathcal{R}$ to the corresponding segments. By definition of a weighted marked pattern count, this is exactly the value of $\mathcal{T}'_u(\mathcal{R}')$. \\

In the remainder of this section,
we design algorithms computing the pattern counts $\pcwe{321\underline{4}}$ and $\pce{4321\underline{5}}$ in sub-quadratic time.
Consequently, we can insert vertices labeled $\mathtt{3214}$ and $\mathtt{43215}$ into pattern-trees of maximum size $1$ and
with at most $5$ points.
Using the above modification to \Cref{alg:bottom_up_pattern_tree}, the overall time complexity for the 
evaluation of such trees remains sub-quadratic.

\subsection{Computing \texorpdfstring{$\pcwe{321\underline{4}}$}{\# 3214} in \texorpdfstring{$\Otilde{n^{5/3}}$}{O(n 5/3)} time}
\label{subsect:3214_gadget}

Theorem 1.2 of \cite{even2021counting} describes an algorithm for counting $\pce{3214}$. 
We require a slight alteration of their algorithm, and in particular, a weighted variant.

\begin{lemma} [weighted version of Theorem 1.2 in \cite{even2021counting}]
    Given an input permutation $\pi \in\Sn$ and weight functions $w_1,\ldots,w_4:[n]\to\ZZ$, the tree $\pcw{321\underline{4}}{\pi}$ can be computed in $\Otilde{n^{5/3}}$ time.
\end{lemma}
\begin{proof}
Let $\pi\in\Sn$ and $w_1,\ldots,w_4$ be as in the statement.
We show how to construct a new rectangle-tree $\mathcal{T}_{out}$,
in which every point $(i,\pi(i))\in p(\pi)$ is weighted according to the
weighted count of $\mathtt{3214}$-occurrences that end in that point.
Let $m\in [n]$ be a parameter (to be chosen later).
Partition $p(\pi)$ into $\lceil n / m \rceil$ non-overlapping horizontal strips,
each of height $m$ except possibly the last one. Perform a similar partition vertically,
with strip width $m$.
Formally, a point $(i,\pi(i))\in p(\pi)$
belongs to the vertical strip indexed $\lceil i/m \rceil$ and to the horizontal strip indexed $\lceil \pi(i)/m \rceil$.
We split to cases with respect to the strips:
in any specific occurrence of $\mathtt{3214}$,
the point $\mathtt{4}$ may or may not share a horizontal strip with $\mathtt{3}$,
and may or may not share a vertical strip with $\mathtt{1}$.

\begin{figure}[ht]
    \centering
    \definecolor{darkgreen}{HTML}{285238}
    \begin{tikzpicture}[scale=0.32]
        \draw[help lines, color=darkgray, opacity=0.8] (0,0) grid (10,10);
        \node[fit={(0,0) (10,10)}, inner sep=0pt, draw=black] (square) {};
        \node[fit={(0,9) (10,6)}, inner sep=0pt, draw=darkgreen, fill=green, fill opacity=0.15] (strip_h) {};
        \node[fit={(6,0) (9,10)}, inner sep=0pt, draw=darkgreen, fill=green, fill opacity=0.15] (strip_v) {};

        \draw [decorate,decoration={brace,amplitude=5pt,mirror,raise=0.7mm}] (0,9) -- (0,6) node[midway,xshift=-5mm]{$m$};
        \draw [decorate,decoration={brace,amplitude=5pt,mirror,raise=0.7mm}] (6,0) -- (9,0) node[midway,yshift=-5mm]{$m$};

        \node[fit={(0,0) (6,6)}, inner sep=0pt, line width=1mm, pattern=north west lines, distance=1500pt, pattern color=blue] (rect) {};
        
        \fill[black] (0.5,5.5) circle (8pt) {};
        \draw[darkgray] (0.5,5.5) circle (8pt) {};

        \fill[black] (2.5,4.5) circle (8pt) {};
        \draw[darkgray] (2.5,4.5) circle (8pt) {};

        \fill[black] (3.5,6.5) circle (8pt) {};
        \draw[darkgray] (3.5,6.5) circle (8pt) {};

        \fill[black] (5.5,0.5) circle (8pt) {};
        \draw[darkgray] (5.5,0.5) circle (8pt) {};

        \fill[black] (4.5,3.5) circle (8pt) {};
        \draw[darkgray] (4.5,3.5) circle (8pt) {};

        \fill[black] (1.5,7.5) circle (8pt) {};
        \draw[darkgray] (1.5,7.5) circle (8pt) {};

        \draw[black] (7.5,8.5)  node {\scriptsize $\mathtt{4}$};
        \draw[blue] (7.5,8.5) circle (12pt) {};

        \fill[black] (6.5,1.5) circle (8pt) {};
        \draw[darkgray] (6.5,1.5) circle (8pt) {};

        \fill[black] (8.5,9.5) circle (8pt) {};
        \draw[darkgray] (8.5,9.5) circle (8pt) {};
        
        \fill[black] (9.5,2.5) circle (8pt) {};
        \draw[darkgray] (9.5,2.5) circle (8pt) {};
    
    \end{tikzpicture}
    \hspace{1cm}
    \begin{tikzpicture}[scale=0.32]
        \draw[help lines, color=darkgray, opacity=0.8] (0,0) grid (10,10);
        \node[fit={(0,0) (10,10)}, inner sep=0pt, draw=black] (square) {};
        \node[fit={(0,9) (10,6)}, inner sep=0pt, draw=darkgreen, fill=green, fill opacity=0.15] (strip_h) {};
        \node[fit={(6,0) (9,10)}, inner sep=0pt, draw=darkgreen, fill=green, fill opacity=0.15] (strip_v) {};

        \draw [decorate,decoration={brace,amplitude=5pt,mirror,raise=0.7mm}] (0,9) -- (0,6) node[midway,xshift=-5mm]{$m$};
        \draw [decorate,decoration={brace,amplitude=5pt,mirror,raise=0.7mm}] (6,0) -- (9,0) node[midway,yshift=-5mm]{$m$};

        \node[fit={(0,0) (7,6)}, inner sep=0pt, line width=1mm, pattern=north west lines, distance=1500pt, pattern color=blue] (rect) {};
        
        \fill[black] (0.5,5.5) circle (8pt) {};
        \draw[darkgray] (0.5,5.5) circle (8pt) {};

        \fill[black] (2.5,4.5) circle (8pt) {};
        \draw[darkgray] (2.5,4.5) circle (8pt) {};

        \fill[black] (3.5,6.5) circle (8pt) {};
        \draw[darkgray] (3.5,6.5) circle (8pt) {};

        \fill[black] (5.5,0.5) circle (8pt) {};
        \draw[darkgray] (5.5,0.5) circle (8pt) {};

        \fill[black] (4.5,3.5) circle (8pt) {};
        \draw[darkgray] (4.5,3.5) circle (8pt) {};

        \fill[black] (1.5,7.5) circle (8pt) {};
        \draw[darkgray] (1.5,7.5) circle (8pt) {};

        \draw[black] (7.5,8.5)  node {\scriptsize $\mathtt{4}$};
        \draw[blue] (7.5,8.5) circle (12pt) {};

        \fill[black] (6.5,1.5) circle (8pt) {};
        \draw[darkgray] (6.5,1.5) circle (8pt) {};

        \fill[black] (8.5,9.5) circle (8pt) {};
        \draw[darkgray] (8.5,9.5) circle (8pt) {};
        
        \fill[black] (9.5,2.5) circle (8pt) {};
        \draw[darkgray] (9.5,2.5) circle (8pt) {};
    
    \end{tikzpicture}
    \hspace{1cm}
    \begin{tikzpicture}[scale=0.32]
        \draw[help lines, color=darkgray, opacity=0.8] (0,0) grid (10,10);
        \node[fit={(0,0) (10,10)}, inner sep=0pt, draw=black] (square) {};
        \node[fit={(0,9) (10,6)}, inner sep=0pt, draw=darkgreen, fill=green, fill opacity=0.15] (strip_h) {};
        \node[fit={(6,0) (9,10)}, inner sep=0pt, draw=darkgreen, fill=green, fill opacity=0.15] (strip_v) {};

        \draw [decorate,decoration={brace,amplitude=5pt,mirror,raise=0.7mm}] (0,9) -- (0,6) node[midway,xshift=-5mm]{$m$};
        \draw [decorate,decoration={brace,amplitude=5pt,mirror,raise=0.7mm}] (6,0) -- (9,0) node[midway,yshift=-5mm]{$m$};

        \node[fit={(2,2) (6,7)}, inner sep=0pt, line width=1mm, pattern=north west lines, distance=1500pt, pattern color=blue] (rect) {};
        
        \fill[black] (0.5,5.5) circle (8pt) {};
        \draw[darkgray] (0.5,5.5) circle (8pt) {};

        \fill[black] (2.5,4.5) circle (8pt) {};
        \draw[darkgray] (2.5,4.5) circle (8pt) {};

        \fill[black] (3.5,6.5) circle (8pt) {};
        \draw[darkgray] (3.5,6.5) circle (8pt) {};

        \fill[black] (5.5,0.5) circle (8pt) {};
        \draw[darkgray] (5.5,0.5) circle (8pt) {};

        \fill[black] (4.5,3.5) circle (8pt) {};
        \draw[darkgray] (4.5,3.5) circle (8pt) {};

        \draw[black] (1.5,7.5) node {\scriptsize $\mathtt{3}$};
        \draw[red] (1.5,7.5) circle (12pt) {};

        \draw[black] (7.5,8.5)  node {\scriptsize $\mathtt{4}$};
        \draw[blue] (7.5,8.5) circle (12pt) {};

        \draw[black] (6.5,1.5) node {\scriptsize $\mathtt{1}$};
        \draw[red] (6.5,1.5) circle (12pt) {};

        \fill[black] (8.5,9.5) circle (8pt) {};
        \draw[darkgray] (8.5,9.5) circle (8pt) {};
        
        \fill[black] (9.5,2.5) circle (8pt) {};
        \draw[darkgray] (9.5,2.5) circle (8pt) {};
    
    \end{tikzpicture}
    
    \vspace{-0.1in}
    \caption{An illustration of the various cases for a given point $\mathtt{4}$ (circled blue), whose strips are highlighted in green. In the case of no shared strips (left), we count the number of descending triplets in the blue-shaded area. To allow $\mathtt{4}$ and $\mathtt{1}$ to share a vertical strip, the area is extended accordingly (centre). Sharing both is depicted on the right:  $\mathtt{3}$ and $\mathtt{1}$ are selected (circled red), and it remains to count the $\mathtt{2}$'s in the blue-shaded area.}
    \label{fig:3214}
\end{figure}
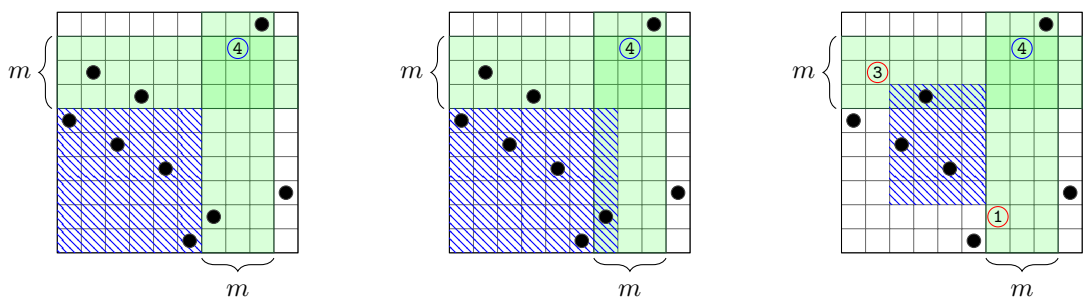

\paragraph{No shared strips.}
We handle each horizontal strip separately.
Let $1\le y \le \lceil n/m\rceil$ and let $(i,\pi(i))$ be a point in the $y$-th horizontal strip, that is, $\lceil \pi(i) / m \rceil = y$. Let $x=\lceil i / m\rceil$ be the index of the point's vertical strip. In order to count the weighted number of $\mathtt{3214}$ occurrences that end in $(i,\pi(i))$
and where $\mathtt{4}$ does not share a horizontal strip with $\mathtt{3}$ nor a vertical strip with $\mathtt{1}$,
we count the weighted number of descending triplets in the rectangle $\mathcal{R}_i\eqdef[1, (x-1) m] \times [1, (y - 1)m]$,
i.e., to the left of the vertical strip and beneath the horizontal strip (see \Cref{fig:3214}).

To do this efficiently, we first consider all points below the $y$-th strip. Construct a rectangle-tree $\mathcal{T}_3$ as in \Cref{prop:count-123k} for descending patterns, except the weight of each point $(i,\pi(i))$ in each of $\mathcal{T}_1$,$\mathcal{T}_2$ and $\mathcal{T}_3$ is multiplied by $w_1(i)$, $w_2(i)$, and $w_3(i)$, respectively. Now, for every point $(i,\pi(i))$ in the strip, we query $\mathcal{T}_3(\mathcal{R}_i)$ to obtain the weighted count of descending triplets below.
Multiply this by $w_4(i)$ and add the result to $(i,\pi(i))$ in $\mathcal{T}_{out}$. There are $\Oh{m}$ points in a strip, so we handle one strip in $\Otilde{n+m}$ time. Repeating for each strip, this case takes $\Otilde{(n + m)n/m} = \Otilde{n^2/m}$ time.

\paragraph{Only sharing vertical strip with $\mathbf{\mathtt{1}}$.}
Let $y,i$ be as above, and repeat the calculation from the previous case.
Observe that querying the rectangle $\mathcal{R} = [1,i-1] \times [1,(y-1)m]$ (i.e., all points to the left of $(i, \pi(i))$
and beneath its vertical strip) counts all triplets in which $\mathtt{4}$ does not share a horizontal strip with $\mathtt{3}$
(and may or may not share a vertical strip with $\mathtt{1}$), see \Cref{fig:3214}.
Subtracting this value from that of the previous case's query gives the desired result.

\paragraph{Only sharing horizontal strip with $\mathbf{\mathtt{3}}$.}
This is a reflection of the previous case along the main diagonal. So, invoke the previous case over the
input permutation $\pi^{-1}$ (i.e., act with $sr^{-1} \in D_4$ as preprocessing, as explained in \Cref{sect:prelims}).

\paragraph{Sharing both strips.}
Iterate over every point $(i,\pi(i))\in p(\pi)$, thinking of each as a $\mathtt{4}$ in the pattern.
Then, iterate over the $\Oh{m}$ points with which it shares a horizontal strip as candidates for the $\mathtt{3}$,
and over the $\Oh{m}$ points with which it shares a vertical strip as candidates for $\mathtt{1}$.
For each pair of such candidates, if they indeed form a descending pair to the bottom-left of $(i,\pi(i))$,
the number of $\mathtt{321}$ contributed by them is exactly the amount of points in the rectangle defined by them
(see \Cref{fig:3214}).
This can be computed with a query to a $2$-dimensional rectangle-tree $\mathcal{T}$ that we construct in preprocessing.
Since these points are candidates for the $\mathtt{2}$ in the pattern,
the weight of every $(j,\pi(j))\in p(\pi)$ in $\mathcal{T}$ is $w_2(j)$.
Multiply the query result by the corresponding weights of the current candidates for $\mathtt{1}$, $\mathtt{3}$ and $\mathtt{4}$, 
and add the result to $(i,\pi(i))$ in $\mathcal{T}_{out}$.
We perform at most $\Oh{m^2}$ queries per point $(i,\pi(i))$, so this case takes $\Otilde{nm^2}$ time.

\paragraph{}
Combining the cases, the complexity is $\Otilde{n^2/m+n m^2}$, which is minimised at $\Otilde{n^{5/3}}$ by fixing $m=\lfloor n^{1/3} \rfloor$.
\end{proof}

\subsection{Computing \texorpdfstring{$\pce{4321\underline{5}}$}{\# 43215} in \texorpdfstring{$\Otilde{n^{7/4}}$}{O(n 7/4)} time}
\label{subsect:43215_gadget}

\begin{lemma}
    \label{lemma:fast_43215}
    Let $\pi\in\Sn$ be a permutation. Then, $\pc{4321\underline{5}}{\pi}$ can be computed in $\Otilde{n^{7/4}}$ time.
\end{lemma}
\begin{proof}

We adapt the $\pcwe{321\underline{4}}$ algorithm to the pattern $\pce{4321\underline{5}}$.
Once again, partition $p(\pi)$ into horizontal and vertical strips of size $m$,
and consider the following possible cases,
corresponding to whether $\mathtt{5}$ shares a vertical strip with $\mathtt{1}$ and/or a horizontal strip with $\mathtt{4}$.

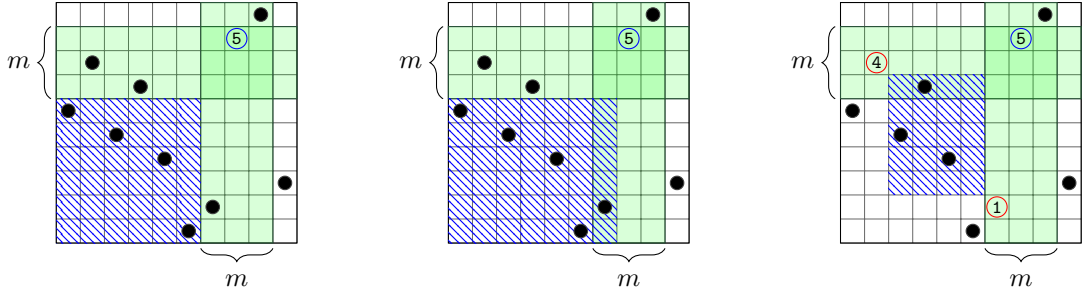
\begin{figure}[ht]
    \centering
    \definecolor{darkgreen}{HTML}{285238}
    \begin{tikzpicture}[scale=0.32]
        \draw[help lines, color=darkgray, opacity=0.8] (0,0) grid (10,10);
        \node[fit={(0,0) (10,10)}, inner sep=0pt, draw=black] (square) {};
        \node[fit={(0,9) (10,6)}, inner sep=0pt, draw=darkgreen, fill=green, fill opacity=0.15] (strip_h) {};
        \node[fit={(6,0) (9,10)}, inner sep=0pt, draw=darkgreen, fill=green, fill opacity=0.15] (strip_v) {};

        \draw [decorate,decoration={brace,amplitude=5pt,mirror,raise=0.7mm}] (0,9) -- (0,6) node[midway,xshift=-5mm]{$m$};
        \draw [decorate,decoration={brace,amplitude=5pt,mirror,raise=0.7mm}] (6,0) -- (9,0) node[midway,yshift=-5mm]{$m$};

        \node[fit={(0,0) (6,6)}, inner sep=0pt, line width=1mm, pattern=north west lines, distance=1500pt, pattern color=blue] (rect) {};
        
        \fill[black] (0.5,5.5) circle (8pt) {};
        \draw[darkgray] (0.5,5.5) circle (8pt) {};

        \fill[black] (2.5,4.5) circle (8pt) {};
        \draw[darkgray] (2.5,4.5) circle (8pt) {};

        \fill[black] (3.5,6.5) circle (8pt) {};
        \draw[darkgray] (3.5,6.5) circle (8pt) {};

        \fill[black] (5.5,0.5) circle (8pt) {};
        \draw[darkgray] (5.5,0.5) circle (8pt) {};

        \fill[black] (4.5,3.5) circle (8pt) {};
        \draw[darkgray] (4.5,3.5) circle (8pt) {};

        \fill[black] (1.5,7.5) circle (8pt) {};
        \draw[darkgray] (1.5,7.5) circle (8pt) {};

        \draw[black] (7.5,8.5)  node {\scriptsize $\mathtt{5}$};
        \draw[blue] (7.5,8.5) circle (12pt) {};

        \fill[black] (6.5,1.5) circle (8pt) {};
        \draw[darkgray] (6.5,1.5) circle (8pt) {};

        \fill[black] (8.5,9.5) circle (8pt) {};
        \draw[darkgray] (8.5,9.5) circle (8pt) {};
        
        \fill[black] (9.5,2.5) circle (8pt) {};
        \draw[darkgray] (9.5,2.5) circle (8pt) {};
    
    \end{tikzpicture}
    \hspace{1cm}
    \begin{tikzpicture}[scale=0.32]
        \draw[help lines, color=darkgray, opacity=0.8] (0,0) grid (10,10);
        \node[fit={(0,0) (10,10)}, inner sep=0pt, draw=black] (square) {};
        \node[fit={(0,9) (10,6)}, inner sep=0pt, draw=darkgreen, fill=green, fill opacity=0.15] (strip_h) {};
        \node[fit={(6,0) (9,10)}, inner sep=0pt, draw=darkgreen, fill=green, fill opacity=0.15] (strip_v) {};

        \draw [decorate,decoration={brace,amplitude=5pt,mirror,raise=0.7mm}] (0,9) -- (0,6) node[midway,xshift=-5mm]{$m$};
        \draw [decorate,decoration={brace,amplitude=5pt,mirror,raise=0.7mm}] (6,0) -- (9,0) node[midway,yshift=-5mm]{$m$};

        \node[fit={(0,0) (7,6)}, inner sep=0pt, line width=1mm, pattern=north west lines, distance=1500pt, pattern color=blue] (rect) {};
        
        \fill[black] (0.5,5.5) circle (8pt) {};
        \draw[darkgray] (0.5,5.5) circle (8pt) {};

        \fill[black] (2.5,4.5) circle (8pt) {};
        \draw[darkgray] (2.5,4.5) circle (8pt) {};

        \fill[black] (3.5,6.5) circle (8pt) {};
        \draw[darkgray] (3.5,6.5) circle (8pt) {};

        \fill[black] (5.5,0.5) circle (8pt) {};
        \draw[darkgray] (5.5,0.5) circle (8pt) {};

        \fill[black] (4.5,3.5) circle (8pt) {};
        \draw[darkgray] (4.5,3.5) circle (8pt) {};

        \fill[black] (1.5,7.5) circle (8pt) {};
        \draw[darkgray] (1.5,7.5) circle (8pt) {};

        \draw[black] (7.5,8.5)  node {\scriptsize $\mathtt{5}$};
        \draw[blue] (7.5,8.5) circle (12pt) {};

        \fill[black] (6.5,1.5) circle (8pt) {};
        \draw[darkgray] (6.5,1.5) circle (8pt) {};

        \fill[black] (8.5,9.5) circle (8pt) {};
        \draw[darkgray] (8.5,9.5) circle (8pt) {};
        
        \fill[black] (9.5,2.5) circle (8pt) {};
        \draw[darkgray] (9.5,2.5) circle (8pt) {};
    
    \end{tikzpicture}
    \hspace{1cm}
    \begin{tikzpicture}[scale=0.32]
        \draw[help lines, color=darkgray, opacity=0.8] (0,0) grid (10,10);
        \node[fit={(0,0) (10,10)}, inner sep=0pt, draw=black] (square) {};
        \node[fit={(0,9) (10,6)}, inner sep=0pt, draw=darkgreen, fill=green, fill opacity=0.15] (strip_h) {};
        \node[fit={(6,0) (9,10)}, inner sep=0pt, draw=darkgreen, fill=green, fill opacity=0.15] (strip_v) {};

        \draw [decorate,decoration={brace,amplitude=5pt,mirror,raise=0.7mm}] (0,9) -- (0,6) node[midway,xshift=-5mm]{$m$};
        \draw [decorate,decoration={brace,amplitude=5pt,mirror,raise=0.7mm}] (6,0) -- (9,0) node[midway,yshift=-5mm]{$m$};

        \node[fit={(2,2) (6,7)}, inner sep=0pt, line width=1mm, pattern=north west lines, distance=1500pt, pattern color=blue] (rect) {};
        
        \fill[black] (0.5,5.5) circle (8pt) {};
        \draw[darkgray] (0.5,5.5) circle (8pt) {};

        \fill[black] (2.5,4.5) circle (8pt) {};
        \draw[darkgray] (2.5,4.5) circle (8pt) {};

        \fill[black] (3.5,6.5) circle (8pt) {};
        \draw[darkgray] (3.5,6.5) circle (8pt) {};

        \fill[black] (5.5,0.5) circle (8pt) {};
        \draw[darkgray] (5.5,0.5) circle (8pt) {};

        \fill[black] (4.5,3.5) circle (8pt) {};
        \draw[darkgray] (4.5,3.5) circle (8pt) {};

        \draw[black] (1.5,7.5) node {\scriptsize $\mathtt{4}$};
        \draw[red] (1.5,7.5) circle (12pt) {};

        \draw[black] (7.5,8.5)  node {\scriptsize $\mathtt{5}$};
        \draw[blue] (7.5,8.5) circle (12pt) {};

        \draw[black] (6.5,1.5) node {\scriptsize $\mathtt{1}$};
        \draw[red] (6.5,1.5) circle (12pt) {};

        \fill[black] (8.5,9.5) circle (8pt) {};
        \draw[darkgray] (8.5,9.5) circle (8pt) {};
        
        \fill[black] (9.5,2.5) circle (8pt) {};
        \draw[darkgray] (9.5,2.5) circle (8pt) {};
    
    \end{tikzpicture}
    
    \vspace{-0.1in}
    \caption{An illustration of the various cases for a given point $\mathtt{5}$ (circled blue). The cases are analogous to $\pcwe{321\underline{4}}$ (see \Cref{fig:3214}). In the case of $\mathtt{5}$ sharing with at most one of $\mathtt{1},\mathtt{4}$ (left and centre), we count descending quadruplets instead of triplets as in $\pcwe{321\underline{4}}$. In the case of sharing both (right),  $\mathtt{4}$ and $\mathtt{1}$ are selected (circled red), and it remains to count descending pairs in the blue-shaded area, corresponding to $\mathtt{32}$.}
    \label{fig:43215}
\end{figure}

\paragraph{No sharing, or sharing with at most one of $\mathtt{1},\mathtt{4}$.}
In the algorithm for $\pcwe{321\underline{4}}$,
we handled both of these cases separately for each horizontal strip, by counting descending triplets.
This was done by first constructing a tree $\mathcal{T}_3$ as in \Cref{prop:count-123k}
to answer such queries for the points strictly below the strip. 
The same can be done for $\pce{4321\underline{5}}$, counting descending \textit{quadruplets} instead.
This does not affect the complexity, as shown in \Cref{prop:count-123k}.
The complexity is $\Otilde{n}$ per strip, totaling $\Otilde{n^2/m}$.

\paragraph{Case of sharing with both.}
The case where $\mathtt{5}$ shares both a vertical strip with $\mathtt{1}$
\textit{and} a horizontal strip with $\mathtt{4}$ is more challenging.
As before, we iterate over all $n$ potential choices of $\mathtt{5}$,
and over all $\Oh{m^2}$ choices of $\mathtt{4}$ and $\mathtt{1}$.
In the algorithm for $\pcwe{321\underline{4}}$,
we counted permutation points in the rectangle defined by $\mathtt{1}$ and $\mathtt{3}$,
whereas here we need to count the number of \emph{descending pairs}
in the rectangle defined by $\mathtt{1}$ and $\mathtt{4}$ (see \Cref{fig:43215}).
In \Cref{thm:pair_rect_tree} below, we construct a data structure that can handle queries of the form ``how many descending pairs are in a given rectangle?''.
The data structure has preprocessing time $\Otilde{n^2/q}$,
and query time $\Otilde{q}$, where $q\in [n]$ is a parameter that can be chosen arbitrarily.
As there are $\Oh{nm^2}$ queries, this case takes $\Otilde{n^2/q}+\Otilde{nm^2 q}$ time in total.

\paragraph{}
Overall, the combined cases take
$\Otilde{n^2/m} + \Otilde{n^2/q} + \Otilde{nm^2 q}$ time, minimised at $\Otilde{n^{7/4}}$ by fixing $m=q=\lfloor n^{1/4} \rfloor$.
\end{proof}

\subsection{Pair-Rectangle-Trees}
\label{subsect:pair_rect_trees}

\begin{theorem}(pair-rectangle-tree)
    \label{thm:pair_rect_tree}
    There exists a data structure with the following properties:
    \begin{enumerate}
        \item \underline{Preprocessing}: Given an input permutation $\pi \in \Sn$, the tree is initialised in time $\Otilde{n^2 / q}$.
        \item \underline{Query}: Given any rectilinear rectangle $\mathcal{R} \subseteq [n] \times [n]$, return the number of descending (resp. ascending) pairs of permutations points in $\mathcal{R}$, in time $\Otilde{q}$.
    \end{enumerate}
    where $q=q(n) \in [n]$ is a parameter that can be chosen arbitrarily. 
\end{theorem}
\begin{proof} We handle each property separately.

    \paragraph{Preprocessing.}
    Construct a $2$-dimensional rectangle-tree $\mathcal{T}_1$ and insert every point in $p(\pi)$ with weight $1$.
    For any rectangle $\mathcal{R}\subseteq [n] \times [n]$,
    the query $\mathcal{T}_1(\mathcal{R})$ counts the number of permutation points in $\mathcal{R}$.
    It is sufficient to consider the case of ascending pairs;
    to count descending pairs in $\mathcal{R}$, subtract the number of ascending pairs from the number of pairs, $\binom{\mathcal{T}_1(\mathcal{R})}{2}$.
    Construct another rectangle-tree $\mathcal{T}_2$ and insert every point $(i,\pi(i))\in p(\pi)$ with the weight $\mathcal{T}_1([1,i-1]\times [1,\pi(i)-1])$,
    where we subtract $1$ to exclude the point itself.
    The query $\mathcal{T}_2(\mathcal{R})$ counts ascending pairs in $\pi$ that end in $\mathcal{R}$.

    Partition $p(\pi)$ into contiguous non-overlapping strips of size $q$, both vertically and horizontally.
    Formally, for every $1\le s\le\lceil n/q \rceil$, define a vertical strip and a horizontal strip:
    \[
    V_{s}\eqdef\left\{ (i,\pi(i)) \in p(\pi) : \left\lceil i/q \right\rceil = s\right\},\;H_{s}\eqdef\left\{ (i,\pi(i))\in p(\pi)  : \left\lceil \pi(i)/q \right\rceil = s\right\}
    \]
    The strips can be constructed in linear time by iterating once over $p(\pi)$ and adding each point to the two appropriate strips.
    Since $\pi$ is a permutation, every strip contains exactly $q$ points,
    except possibly for the two corresponding to $s=\lceil n / q \rceil$, which may be smaller.
    
    For every vertical strip $V_s$, we construct a $2$-dimensional rectangle-tree $\mathcal{T}^V_s$.
    This tree is similar to $\mathcal{T}_2$, but only takes into account ascending pairs that start in $V_s$ or to its left.
    Formally, we insert every point $(i,\pi(i))\in p(\pi)$ into $\mathcal{T}^V_s$ with weight
    \[
    \mathcal{T}_1([1,\min(i-1, s\cdot q)] \times [1,\pi(i)-1])
    \]
    Symmetrically, for every horizontal strip $H_s$ we construct a tree $\mathcal{T}^H_s$,
    which allows us to count how many ascending pairs end in a given rectangle and start in or below $H_s$.
    Overall, we construct $\Oh{n/q}$ rectangle-trees, at cost $\Otilde{n}$ each, totaling $\Otilde{n^2/q}$ preprocessing time.

    \paragraph{Query.}
    Let $\mathcal{R} = [x_1,x_2] \times [y_1,y_2]$.
    Consider the outermost strips that $\mathcal{R}$ may overlap,
    corresponding to indices $a\eqdef \lceil x_1/q\rceil$, $b\eqdef \lceil x_2/q\rceil$, $c\eqdef \lceil y_1/q\rceil$, and $d\eqdef \lceil y_2/q\rceil$. Define the \emph{margin} $M\subseteq \mathcal{R}$ as the set of permutation points contained both in $\mathcal{R}$ and in the outermost strips:
    \[
    M\eqdef \mathcal{R} \cap \left( V_a \cup V_b \cup H_c \cup H_d\right)
    \]
    Define the \emph{interior} $\mathcal{R}_{in}\subseteq \mathcal{R}$ as the rectangle obtained by trimming the margin. Formally,
    \[
    \mathcal{R}_{in} \eqdef \left\{ (x,y)\in [n]\times [n] : a < \lceil x / q\rceil  < b \text{ and } c < \lceil y / q \rceil < d \right\}
    \]
    See \Cref{fig:pair-tree}. Note that $\mathcal{R}$ is possibly contained in a single strip or in two consecutive strips, in which case its interior is empty. \ \\

    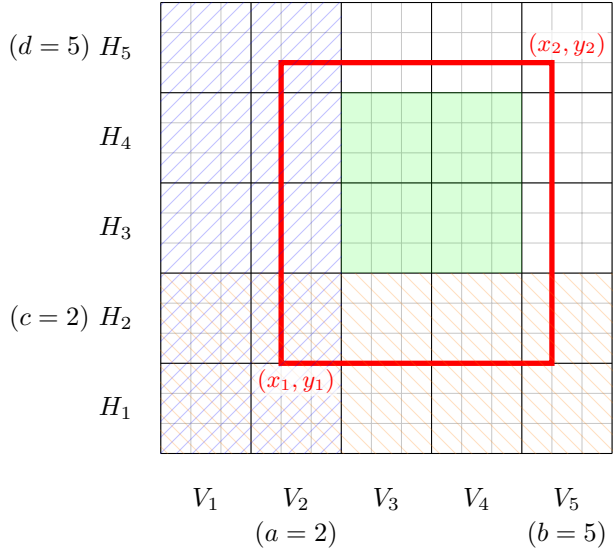
\begin{figure}[ht]
        \centering
        \begin{tikzpicture}[scale=0.4]
             \newcommand\xymax{15};
            \draw[help lines, color=darkgray, opacity=0.3] (0,0) grid (\xymax,\xymax);
             \foreach \i\s in {0/1,3/2,6/3,9/4,12/5} {
                 \draw[black] (\i+1.5,-1.5) node {$V_{\s}$};
                 \draw[black] (-1.5,\i+1.5) node {$H_{\s}$};
             }
             \foreach \i in {0,3,...,\xymax} {
                 \draw [black] (\i,0) -- (\i,\xymax);
                 \draw [black] (0,\i) -- (\xymax,\i);
             }
             
             \newcommand\rxl{4};
             \newcommand\rxr{13};
             \newcommand\ryb{3};
             \newcommand\ryt{13};
    
             \newcommand\rinxl{6};
             \newcommand\rinxr{12};
             \newcommand\rinyb{6};
             \newcommand\rinyt{12};
    
             \draw[black] (-3.7,3+1.5) node {$(c=2)$};
             \draw[black] (-3.7,\rinyt+1.5) node {$(d=5)$};
             \draw[black] (\rinxl-1.5,-2.7) node {$(a=2)$};
             \draw[black] (13.5,-2.7) node {$(b=5)$};
             
             \node[fit={(\rinxl,\rinyb) (\rinxr,\rinyt)}, inner sep=0pt, line width=1mm, fill=green, fill opacity=0.15] (rect) {};
             \node[fit={(\rxl,\ryb) (\rxr,\ryt)}, inner sep=0pt, line width=2pt, draw=red, opacity=1] (rect) {};
    
            \node[fit={(0,0) (\rinxl,\xymax)}, inner sep=0pt, line width=1mm,pattern=south west lines, pattern color=blue, opacity=0.5] (rect) {};
    
            \node[fit={(0,0) (\xymax,\rinyb)}, inner sep=0pt, line width=1mm, pattern=south east lines, pattern color=orange, opacity=0.5] (rect) {};
    
             \node [fill=white, fill opacity=0.7, text=red, text opacity=1, rounded corners=2pt,inner sep=1pt] at (\rxr+0.5,\ryt+0.6) {\small $(x_2,y_2)$};
        
             \node [fill=white, fill opacity=0.7, text=red, text opacity=1, rounded corners=2pt,inner sep=1pt] at (\rxl+0.5,\ryb-0.6) {\small $(x_1,y_1)$};
        
        \end{tikzpicture}
        \caption{Illustration of a query rectangle $\mathcal{R}=[5,13]\times [4,13]$, marked with a red border. The size of every strip is $q=3$. The indices of $\mathcal{R}$'s outermost strips are $a=2$, $b=5$, $c=2$ and $d=5$. The interior $\mathcal{R}_{in}$ is highlighted in green, and does not overlap any of the outermost strips. The blue- and orange-shaded areas are, respectively, the areas left of and below $\mathcal{R}_{in}$. These areas intersect at the bottom-left rectangle, denoted $\mathcal{R}_0$.}
        \label{fig:pair-tree}
    \end{figure}

    We first construct a list of all permutation points in the margin, $M$.
    Since $M$ is contained in the union of $4$ strips, it contains at most $4q$ points.
    These points can be collected in $\Otilde{q}$ time by iterating over all points in the strips, and adding to the list those that are contained in $\mathcal{R}$. \\

    The ascending pairs in $\mathcal{R}$ are counted by splitting to cases.
    For every pair of points in $\mathcal{R}$, there are three possibilities:
    either they end in the margin $M$, or they start in $M$ and end in $\mathcal{R}_{in}$, or they both start and end in $\mathcal{R}_{in}$.

    \subparagraph{Ascending pairs that end in $M$.}
    Iterate over all points $(i,\pi(i))\in M$. The number of ascending pairs in $\mathcal{R}$ that end in $(i,\pi(i))$ is obtained by the query $\mathcal{T}_1([x_1,i-1] \times [y_1, \pi(i)-1])$.

    \subparagraph{Ascending pairs that start in $M$ and end in $\mathcal{R}_{in}$.}
    Iterate over all points $(i,\pi(i))\in M$. The points lying above and to its right
    are contained in the rectangle $[i+1,x_2] \times [\pi(i)+1,y_2]$.
    As we are only interested in pairs that end in $\mathcal{R}_{in}$,
    we perform the query $\mathcal{T}_1(([i+1,x_2]\times [\pi(i)+1,y_2])\cap \mathcal{R}_{in})$.

    \subparagraph{Ascending pairs in $\mathcal{R}_{in}$.}
    This case is different from the previous ones, as there may be more than $\Otilde{q}$ permutation points in $\mathcal{R}_{in}$.
    To avoid iterating over them, we use the constructed rectangle-trees to perform inclusion-exclusion, as follows.
    The query $\mathcal{T}_2(\mathcal{R}_{in})$ counts how many ascending pairs end in $\mathcal{R}_{in}$.
    In addition to our intended purpose, this also counts pairs that end in $\mathcal{R}_{in}$ but start below or to the left of it.
    The number of pairs that start to the left of $\mathcal{R}_{in}$ is $\mathcal{T}^V_a(\mathcal{R}_{in})$,
    and similarly the number of pairs that start below $\mathcal{R}_{in}$ is $\mathcal{T}^H_c(\mathcal{R}_{in})$ (see \Cref{fig:pair-tree}). After subtracting both, we must add back the number of ascending pairs that start below \emph{and} to the left of $\mathcal{R}_{in}$.
    Such pairs start in $\mathcal{R}_0 \eqdef [1, aq] \times [1, cq]$.
    Notice that each point in $\mathcal{R}_{in}$ creates an ascending pair with each point in $\mathcal{R}_0$,
    so the number of such pairs is the product of their point counts.
    Overall, this case contributes the following to the query result:
    \[
    \underbrace{\mathcal{T}_2(\mathcal{R}_{in})}_{\text{\small ending in $\mathcal{R}_{in}$}}
    -
    \underbrace{\mathcal{T}^V_a(\mathcal{R}_{in})}_{\text{\small starting below}}
    -
    \underbrace{\mathcal{T}^H_c(\mathcal{R}_{in})}_{\text{\small starting to left}}
    +
    \underbrace{\mathcal{T}_1(\mathcal{R}_{in}) \cdot \mathcal{T}_1(\mathcal{R}_0)}_{\text{\small starting bottom-left}} \qedhere
    \]
\end{proof}

\subsection{Algorithm for the \texorpdfstring{$5$}{5}-Profile}

\begin{proof} [Proof of \Cref{thm:fast-5-prof}]
    The proof proceeds along the same lines as \Cref{thm:fast_le_7_profile}, for a different family of pattern-trees.
    Let $\mathbb{S}\eqdef \bigsqcup_{k=1}^5 \mathbb{S}_k$.
    Extend the pattern-trees of maximum size $1$ and over no more than $5$ points,
    by allowing the new vertices described in \Cref{subsect:bottom-up-improve}.
    By computer enumeration, there exists a family of $\sum_{k=1}^5 k!=153$ linearly-independent vectors over $\QQ^\mathbb{S}$,
    obtained from the vectors trees,
    along with their orbit under the action of $D_4$ on the symmetric group (see \Cref{sect:prelims}).
    The proof now follows, similarly to \Cref{thm:fast_le_7_profile},
    and we remark that the evaluation of each pattern-tree over $\pi$
    takes at most $\Otilde{n^{7/4}}$ time, as that is the maximum amount of time spent handling any single vertex.
\end{proof}
    
\pagebreak

\section{Discussion}
Some immediate extensions of this work, such as the application of our methods to the $9$-profile,
or the use of pattern-trees with maximum size $3$,
are computationally difficult and likely require further analysis or a different approach (say, algebraic).
Several interesting open questions remain:

\begin{enumerate}
    \item \textbf{Maximum size versus rank.}
    For a given integer $s$,
    denote by $f(s)$ the largest integer $k$ such that the subspace spanned by the vectors of pattern-trees over at most $k$ points,
    of maximum size $s$ and with no equalities, is of full dimension, $|\mathbb{S}_{\le k}|$.
    The results of \cite{even2021counting} imply that $f(1)=3$.
    In \Cref{subsect:5_to_7_prof} and \Cref{subsect:s_8_not_spanned} we prove $f(2)=7$,
    and in \Cref{subsect:k_over_2_alg} we show that $f(s)\ge 2s$, for every $s \ge 1$.
    What is the behavior of $f(s)$?
    For example, do we have $f(s)\ge 4s\pm o(s)$, 
    as attained by the technique of \cite{berendsohn2021finding}?

    \item \textbf{A fine-grained variant of $f(s)$.} For integers $s$ and $k$, let $g(s,k)$ be the number of linearly independent vectors in $\QQ^{\mathbb{S}_k}$ generated by \Cref{alg:bottom_up_pattern_tree} when applied to trees of maximum size $s$ over $k$ permutation points with no equalities. What is the general behavior of $g(s,k)$? This generalises a question of \cite{even2021counting} about corner-trees. The following values are presently known.
    \begin{table}[htbp]
    \begin{center}
    \begin{tabular}{l|cccccccc}
     \diagbox[innerleftsep=0.25cm,innerrightsep=0.08cm,innerwidth=0.5cm,height=0.6cm]{$s$}{$k$}
     & $1$ & $2$ & $3$ & $4$ & $5$ & $6$ & $7$ & $8$\tabularnewline
    \hline 
    $1$ & 1 & 2 & 6 & 23 & $100$ & $463$ & $2323$ & $\mathbf{12173}$ \tabularnewline
    $2$ & $\mathbf{1}$ & $\mathbf{2}$ & $\textbf{6}$ & $\mathbf{24}$ & $\mathbf{120}$ & $\mathbf{720}$ & $\mathbf{5040}$ & $\mathbf{40319}$ \tabularnewline
    \end{tabular}
    \caption{Bolded values in this table are computed in this paper (new).}
    \end{center}
    \end{table}

    \item \textbf{Complexity of $k$-profile, for $5 \le k \le 7$.} Can the time complexity for finding the $5,6,7$-profiles be improved further, perhaps by utilising techniques along the lines of \Cref{sect:5_prof_alg}? In particular, we ask whether the $6$-profile can be computed in sub-quadratic time.

    \item \textbf{Study of the $8$-profile.} \cite{dudek2020counting} shows the equivalence between the computation of the $4$-profile and counting $4$-cycles in sparse graphs.
    In \Cref{subsect:s_8_not_spanned} we show that many of the observations of \cite{dudek2020counting} can
    be extended to $\mathbb{S}_8$.
    In fact, we conjecture that there exists an analogous hardness result for $k=8$,
    and we refer the reader to \Cref{subsect:s_8_not_spanned}
    where the details are discussed.
\end{enumerate}

\label{sect:discussion}

\bibliography{main}

\begin{thebibliography}{WWWY14}

\bibitem[AAAH01]{albert2001algorithms}
Michael~H Albert, Robert~EL Aldred, Mike~D Atkinson, and Derek~A Holton.
\newblock Algorithms for pattern involvement in permutations.
\newblock In {\em Algorithms and Computation: 12th International Symposium,
  ISAAC 2001 Christchurch, New Zealand, December 19--21, 2001 Proceedings 12},
  pages 355--367. Springer, 2001.

\bibitem[AR08]{ahal2008complexity}
Shlomo Ahal and Yuri Rabinovich.
\newblock On complexity of the subpattern problem.
\newblock {\em SIAM Journal on Discrete Mathematics}, 22(2):629--649, 2008.

\bibitem[AYZ97]{alon1997finding}
Noga Alon, Raphael Yuster, and Uri Zwick.
\newblock Finding and counting given length cycles.
\newblock {\em Algorithmica}, 17(3):209--223, 1997.

\bibitem[Bar68]{bareiss1968sylvester}
Erwin~H Bareiss.
\newblock Sylvester’s identity and multistep integer-preserving gaussian
  elimination.
\newblock {\em Mathematics of computation}, 22(103):565--578, 1968.

\bibitem[BBL98]{bose1998pattern}
Prosenjit Bose, Jonathan~F Buss, and Anna Lubiw.
\newblock Pattern matching for permutations.
\newblock {\em Information Processing Letters}, 65(5):277--283, 1998.

\bibitem[BKM21]{berendsohn2021finding}
Benjamin~Aram Berendsohn, L{\'a}szl{\'o} Kozma, and D{\'a}niel Marx.
\newblock Finding and counting permutations via csps.
\newblock {\em Algorithmica}, 83:2552--2577, 2021.

\bibitem[BLL23]{beniaminilavee2023balanced}
Gal Beniamini, Nir Lavee, and Nati Linial.
\newblock How balanced can permutations be?
\newblock {\em arXiv preprint arXiv:2306.16954}, 2023.

\bibitem[CE87]{chazelle1987linear}
Bernard Chazelle and Herbert Edelsbrunner.
\newblock Linear space data structures for two types of range search.
\newblock {\em Discrete \& Computational Geometry}, 2:113--126, 1987.

\bibitem[Cha88]{chazelle1988functional}
Bernard Chazelle.
\newblock A functional approach to data structures and its use in
  multidimensional searching.
\newblock {\em SIAM Journal on Computing}, 17(3):427--462, 1988.

\bibitem[CP08]{cooper2008symmetric}
Joshua Cooper and Andrew Petrarca.
\newblock Symmetric and asymptotically symmetric permutations.
\newblock {\em arXiv preprint arXiv:0801.4181}, 2008.

\bibitem[DG20]{dudek2020counting}
Bart{\l}omiej Dudek and Pawe{\l} Gawrychowski.
\newblock Counting 4-patterns in permutations is equivalent to counting
  4-cycles in graphs.
\newblock In {\em 31st International Symposium on Algorithms and Computation
  (ISAAC 2020)}. Schloss Dagstuhl-Leibniz-Zentrum f{\"u}r Informatik, 2020.

\bibitem[DP89]{dechter1989tree}
Rina Dechter and Judea Pearl.
\newblock Tree clustering for constraint networks.
\newblock {\em Artificial Intelligence}, 38(3):353--366, 1989.

\bibitem[DWZ23]{duan2023faster}
Ran Duan, Hongxun Wu, and Renfei Zhou.
\newblock Faster matrix multiplication via asymmetric hashing.
\newblock In {\em 2023 IEEE 64th Annual Symposium on Foundations of Computer
  Science (FOCS)}, pages 2129--2138. IEEE, 2023.

\bibitem[ES35]{erdos1935combinatorial}
Paul Erd{\"o}s and George Szekeres.
\newblock A combinatorial problem in geometry.
\newblock {\em Compositio mathematica}, 2:463--470, 1935.

\bibitem[EZ20]{even2020patterns}
Chaim Even-Zohar.
\newblock Patterns in random permutations.
\newblock {\em Combinatorica}, 40(6):775--804, 2020.

\bibitem[EZL21]{even2021counting}
Chaim Even-Zohar and Calvin Leng.
\newblock Counting small permutation patterns.
\newblock In {\em Proceedings of the 2021 ACM-SIAM Symposium on Discrete
  Algorithms (SODA)}, pages 2288--2302. SIAM, 2021.

\bibitem[Fox13]{fox2013stanley}
Jacob Fox.
\newblock Stanley-wilf limits are typically exponential.
\newblock {\em arXiv preprint arXiv:1310.8378}, 2013.

\bibitem[Fre90]{freuderl1990complexity}
Eugene~C. Freuder.
\newblock Complexity of k-tree structured constraint satisfaction problems.
\newblock In {\em Proceedings of the Eighth National Conference on Artificial
  Intelligence - Volume 1}, AAAI'90, page 4–9. AAAI Press, 1990.

\bibitem[GM14]{guillemot2014finding}
Sylvain Guillemot and D{\'a}niel Marx.
\newblock Finding small patterns in permutations in linear time.
\newblock In {\em Proceedings of the twenty-fifth annual ACM-SIAM symposium on
  Discrete algorithms}, pages 82--101. SIAM, 2014.

\bibitem[JMS05]{jaja2005space}
Joseph J{\'a}J{\'a}, Christian~W Mortensen, and Qingmin Shi.
\newblock Space-efficient and fast algorithms for multidimensional dominance
  reporting and counting.
\newblock In {\em Algorithms and Computation: 15th International Symposium,
  ISAAC 2004, Hong Kong, China, December 20-22, 2004. Proceedings 15}, pages
  558--568. Springer, 2005.

\bibitem[Knu97]{knuth1997art}
Donald~E Knuth.
\newblock {\em The Art of Computer Programming: Fundamental Algorithms, Volume
  1}.
\newblock Addison-Wesley Professional, 1997.

\bibitem[Mac15]{macmahon1915combinatory}
Percy~A MacMahon.
\newblock {\em Combinatory analysis, volumes I and II}, volume 137.
\newblock American Mathematical Society, 1915.

\bibitem[MT04]{marcus2004excluded}
Adam Marcus and G{\'a}bor Tardos.
\newblock Excluded permutation matrices and the {S}tanley--{W}ilf conjecture.
\newblock {\em Journal of Combinatorial Theory, Series A}, 107(1):153--160,
  2004.

\bibitem[Pra73]{pratt1973computing}
Vaughan~R Pratt.
\newblock Computing permutations with double-ended queues, parallel stacks and
  parallel queues.
\newblock In {\em Proceedings of the fifth annual ACM symposium on Theory of
  computing}, pages 268--277, 1973.

\bibitem[SS85]{simion1985restricted}
Rodica Simion and Frank~W Schmidt.
\newblock Restricted permutations.
\newblock {\em European Journal of Combinatorics}, 6(4):383--406, 1985.

\bibitem[WWWY14]{williams2014finding}
Virginia~Vassilevska Williams, Joshua~R Wang, Ryan Williams, and Huacheng Yu.
\newblock Finding four-node subgraphs in triangle time.
\newblock In {\em Proceedings of the twenty-sixth annual ACM-SIAM symposium on
  discrete algorithms}, pages 1671--1680. SIAM, 2014.

\end{thebibliography}
\bibliographystyle{alpha}
\appendix
\section{Enumeration of Pattern-Tree Vectors}
\label{sect:enumeration_pattern_trees}

Let $s$ and $k$  be two positive integers, where $s \le k$. 
Consider the following enumeration process, which computes the matrix
whose rows are the vectors of all pattern-trees  of maximum size $\le s$, with exactly $k$ points, restricted to $\mathbb{S}_k$.

For every ordered partition $\lambda \vdash k$ with no part larger than $s$, and 
for every vertex-labeled tree $T \in \mathbb{T}_{|\lambda|}$,\footnote{Here $\mathbb{T}_r$ is the set of all vertex-labeled trees over $r$ vertices.}
let $T_\lambda$ be the tree in which vertex $i$ has size $\lambda(i)$,
and is assigned point variables 
\[
    p(v_i) = \left\{ p_{r_{i-1} + 1}, \dots, p_{r_i} \right\}, \text{ where } r_i \eqdef \sum_{j \le i} \lambda(j), \text{ and } r_0 \eqdef 0.
\]

Think of $T_\lambda$ as a ``template'' for a pattern-tree, where the topology, the sizes of vertices,
and the names of their variables have been determined, but the edge-constraints have not.
Next, iterate over all pairs of permutations, $\sigma, \tau \in \mathbb{S}_k$, and over all trees $T_\lambda$.

Any such combination maps to a pattern-tree in the above family.
For every $i$ and $j$ such that $p_i$ and $p_j$ are associated with the same vertex $v$,
write the constraint $x_i < x_j$ in the vertex $v$ if $\sigma(i) < \sigma(j)$, and write $x_i > x_j$ otherwise.
Do likewise for the $y$ constraints and $\tau$, and repeat the same operation for every pair $i$, $j$ such that
the points $p_i$ and $p_j$ are associated with adjacent vertices in $T_{\lambda}$ --
in this case, we write the inequality on the edge.
Observe that the ordering of points in each vertex is fully determined, i.e., defines a \textit{permutation}.

Therefore, for any combination of $T_{\lambda}$, $\sigma$ and $\tau$, we obtain a pattern-tree $T$ for which
$\sigma$ is a linear extension of the $x$-poset, and $\tau$ is a linear extension of the $y$-poset.
In this case, we add $1$ to the vector of $T$, at the index of $\tau \sigma^{-1} $ (see \Cref{lem:pattern_tree_vec}).
Once the process is completed, we obtain a matrix with $|\mathbb{S}_k|=k!$ columns, and no more than
$|\mathbb{S}_k \times \mathbb{S}_k \times \{ T_\lambda \}|$ rows.
The row-space of this matrix over the rationals
is the subspace spanned by the above family of trees, restricted to $\mathbb{S}_k$.

We remark that if for every $k' \le k$ this process produces a matrix of full rank,
then by induction, the vectors of the union of all trees in these families spans the entire subspace,
for $\mathbb{S}_{\le k}$
(no tree over $k' < k$ points has a component in $\mathbb{S}_k$).

\end{document}